\documentclass[natbib,runningheads]{svjour3}
\usepackage{amsfonts} 
\usepackage[misc]{ifsym}
\usepackage{graphicx}
\usepackage{amsmath}
\usepackage{url} 
\usepackage{bm}
\usepackage{color}
\usepackage{comment}
\usepackage{geometry}

\def\T{{ \mathrm{\scriptscriptstyle T} }}
\newcommand{\tabincell}[2]{\begin{tabular}{@{}#1@{}}#2\end{tabular}}
\def\v{\varepsilon}

\def\bbeta{\bm \beta}

\def\btheta{\bm \theta}

 \def\*#1{\mathbf{#1}}

\begin{document}

\title{Broken Adaptive Ridge Regression for Right-Censored Survival Data
\thanks{The research of Gang Li was partly supported by National Institute of Health Grants P30 CA-16042, P50 CA211015, and UL1TR000124-02.
The research of Zhihua Sun was partly supported by Natural Science Foundation of China 11871444.
The research of Yi Liu was partly supported by Natural Science Foundation of China 11801567.
}}
\vspace{0.8pc}
\centerline{\large\bf Broken Adaptive Ridge Regression for Right-Censored Survival Data}
\vspace{.4cm}
\centerline{Zhihua Sun$^{1}$, Yi Liu$^{1}$, Kani Chen$^{2}$, Gang Li$^{3}$$^{(\textrm{\Letter})}$ 
}
\vspace{.4cm}
\centerline{\it $^{1}$Ocean University of China, $^{2}$Hong Kong University of Science and Technology and}
\centerline{\it $^{3}$University of California at Los Angeles}
\vspace{.55cm}

\institute{ Zhihua Sun \at
              Department of Mathematics, Ocean University of China, Qingdao, China. \\  \email{zhihuasun@ouc.edu.cn}
           \and Yi Liu \at
              Department of Mathematics, Ocean University of China, Qingdao, China.
              \\ \email{liuyi@amss.ac.cn}
             \and Kani Chen \at
  	 Department of Mathematics, Hong Kong University of Science and Technology, Hong Kong. \\ \email{makchen@ust.hk}
            \and \Letter Gang Li \at
              Professor of  Biostatistics and Computational Medicine, University of California, Los Angeles, CA 90095-1772, USA.
              \\ \email{vli@ucla.edu}
}



\begin{abstract}
Broken adaptive ridge (BAR) is a computationally scalable surrogate to $L_0$-penalized regression, which involves iteratively performing reweighted $L_2$ penalized regressions and enjoys some appealing properties of both $L_0$ and $L_2$ penalized regressions while avoiding some of their limitations.
In this paper, we extend the BAR method to the semi-parametric accelerated failure time (AFT) model for right-censored survival data. Specifically, we propose a censored BAR (CBAR) estimator by applying the BAR algorithm to the Leurgan's synthetic data and show that the resulting CBAR estimator is consistent for variable selection, possesses an oracle property for parameter estimation {and enjoys a grouping property for highly correlation covariates}. Both low and high dimensional covariates are considered. The effectiveness of our method is demonstrated and compared  with some popular penalization methods using simulations. Real data illustrations  are provided on a diffuse large-B-cell lymphoma data {and a glioblastoma multiforme data}.
\end{abstract}

\section{Introduction}
\label{section1}
 $L_0$-penalized regression, which
directly penalizes the
cardinality of a model, has been commonly used  for variable selection {in the low dimensional setting}   via well-known information criteria such as Mallow's $C_p$ \citep{mallows1973some}, Akaike's information
criterion (AIC) \citep{akaike1974new}, the Bayesian information criterion (BIC) \citep{schwarz1978estimating,chen2008extended},
and risk inflation criteria (RIC) \citep{foster1994risk}. It has also been
shown to possess some optimal properties for variable selection and parameter estimation \citep{shen2012likelihood, lin2010risk}.
However, $L_0$-penalization is also known to have some limitations such as being computationally NP-hard,  not scalable to  high dimensional data, and unstable for variable selection \citep{breiman1996heuristics}.
To overcome these shortcomings, the broken adaptive ridge (BAR) method \citep{ dai2018fusedBAR, DAI2018334} has been recently introduced as a surrogate to $L_0$ penalization for simultaneous variable selection and parameter estimation  under the linear model \citep{dai2018fusedBAR, DAI2018334}. {It was noted by \citet{dai2018fusedBAR, DAI2018334} that the BAR estimator, defined as the limit of an iteratively reweighted $L_2$ (ridge) penalization algorithm,  retains some appealing properties of $L_0$ penalization while avoiding its pitfalls. For instance,  BAR  generally yields a more sparse, accurate, and  interpretable model  than some popular $L_1$-type penalization methods such as LASSO and its various variations, while maintaining
comparable prediction performance.
Moreover, unlike the exact $L_0$ penalization, BAR is computationally scalable to high dimensional covariates and is stable for variable selection. Lastly,  in addition to being consistent for variable selection and oracle for  parameter estimation, the BAR estimator enjoys a grouping property for highly correlated covariates, a desirable feature not shared by most other oracle variable selection procedures.}


{Because of its appealing properties, the BAR penalization method has  been recently extended to the \citet{Cox1972} model with censored survival data \citep{Eric,Jianguo2019} via penalized likelihood. However, it is well known that the \citet{Cox1972} proportional hazards assumption do not always hold in practice. Thus it is desirable to extend the BAR penalization method to other common survival regression models.
This paper studies an extension of  the BAR penalization method to the semi-parametric accelerated failure time (AFT)  model, a popular alternative to the Cox model for right censored survival data.  To this end, we note that the semi-parametric AFT model is a linear model for the log-transformed survival time with a completely unspecified error distribution, for which the likelihood approach does not yield a consistent parameter estimator even for the classical uncensored linear regression model. Hence,  the BAR penalized likelihood methods of \citet{Eric} and  \citet{Jianguo2019} for the \citet{Cox1972} do not apply to the semiparmetric AFT model. A different approach would be required.}

In this paper, we propose an extension of the BAR  penalization method to the semi-parametric AFT model by coupling the \citet{Leurgans1987} synthetic data approach with the BAR penalty,  study its large sample properties, and demonstrate it effectiveness in comparison with some popular penalization methods using simulations.
  Specifically,  we first use
the \citet{Leurgans1987} synthetic variable method to construct a synthetic outcome variable and then apply the BAR method  for uncensored linear regression \citep{ DAI2018334} to  the synthetic outcome variable. We then give sufficient conditions under which the proposed censored BAR (CBAR) estimator is consistent for variable selection, behaves asymptotically  as well as the oracle estimator based on the true reduced model, {and possesses a grouping property for  highly correlated covariates. We also combine BAR with a sure joint screening method to obtain a two-step variable selection and parameter estimation method for ultra-high dimensional covariates. Not surprisingly, our simulations demonstrate that  the proposed CBAR method generally yields a more sparse and more accurate model as compared to some other popular penalization methods such as LASSO, SCAD, MCP, and adaptive LASSO within the \citet{Leurgans1987} synthetic data framework, which is consistent with the findings of \citet{DAI2018334} for uncensored data. Lastly, we have implemented the proposed CBAR method in an R package, named CenBAR, and made it publicly available at  https://CRAN.R-project.org/package=CenBAR.}

{Before going further, we note that there exist a number of other variable selection methods in the literature for the semiparametric the AFT model with right censored data. These methods are derived by combining various penalization methods such as LASSO with different extensions of the least squares principle for right censored data.
  For example, the Lasso, bridge, elastic net or MCP penalties have been combined with the \citet{stute1993consistent} weighted least squares method \citep{Huang2005,Huang2010Variable,Datta2007Predicting};
    and    the Dantzig, elastic net, Lasso, adaptive Lasso and SCAD penalties  have been combined with the \citet{Buckley1979Linear} method \citep{Yi2009Dantzig,Wang2008Doubly,Johnson2008Penalized,Johnson2009On}.
This paper makes a unique theoretical contribution  since neither the BAR penalization nor the \citet{Leurgans1987} synthetic data method has been previously rigorously studied in the context of variable selection for the semiparametric the AFT model. We also illustrate and compare empirically the  BAR penalization versus some popular penalization methods when the \citet{Leurgans1987} synthetic data least squares method is used.  We do not  compare different penalization methods when they are coupled with different censored least squares methods because different censored least squares methods are derived under different conditions and  none is expected to dominate another across all scenarios.
\par The rest of the paper is organized as follows. In Section 2, we define our CBAR estimator and state its theoretical properties. We also discuss how to handle ultra-high dimensional covariates. In Section \ref{Numerical}, we evaluate the finite sample performance of CBAR in comparison with other penalization methods via extensive  simulations. In Section 4, we illustrate the CBAR method on
 a diffuse large-B-cell lymphoma data {and a glioblastoma multiforme data} with high dimensional covariates. Proofs of the theoretical results are provided in the appendix.

\section{Censored broken adaptive ridge (CBAR) regression }
\label{section2}
\subsection{Notations and preliminaries}
\subsubsection{Model and data}
Consider the linear regression model
\begin{equation}
\label{model:1}
Y_i= \*x_i^\top \bm \beta+\varepsilon_i, \quad  i=1,2,...,n,
\end{equation}
where  for the $i$th subject, $Y_i$ denotes the response variable, $\*x_i$ is the $p_n$-vector random covariates, $\bm\beta = (\beta_1,...,\beta_{p_n})^\top$ is a vector of regression coefficients,  and $\varepsilon_i$  is {i.i.d.} error term with an unknown error distribution, $E(\varepsilon_i)=0$ and $Var(\varepsilon_i)=\sigma^2<\infty$.  Model (\ref{model:1}) is commonly referred to as the accelerated failure time (AFT) model when $Y$ is  the log-transformed survival  time \citep{Kalbfleisch2002}.

Without loss of generality, assume that {$ \bbeta_{0}=( \bbeta_{01}^\top , \bbeta_{02}^{\top})^\top $} is the true value of $\bbeta$, where
$ \bbeta_{01}$ is a ${q}\times 1$ nonzero vector and $ \bbeta_{02}$ is a $(p_n-{q})\times 1$ zero vector. We further assume the columns of the design matrix {$\*X= (\*x_1,...,\*x_n)^\top$} have mean zero and unit $L_2$-norm. Throughout the paper, $\|\cdot\|$ represents the Euclidean norm for a vector and spectral norm for a matrix.

Assume that one observes a right censored data consisting of $n$ independent and identically distributed triples $(T_i,\delta_i,\*x_i)$, $i=1,\ldots, n$, where for the $i$th subject, $T_i=\min(Y_i,C_i)$ is the observation time, $\delta_i=I(Y_i \leq C_i)$ is a
censoring indicator, $C_i$ is the {i.i.d.} censoring time with the distribution function $H$. $C_i$ is assumed to  be independent of $Y_i$ {and $\*x_i$.}

\subsubsection{Broken adaptive ridge (BAR) for uncensored data}
For reader convenience, we first briefly review the broken adaptive ridge (BAR) estimator of \citet{DAI2018334} for
simultaneous variable selection and parameter estimation with the uncensored data
$\*Y$ and $\*X$, where {$\*Y=(Y_1,...,Y_n)^\top$}.

Following the notations of \citet{DAI2018334}, the BAR estimator of $\bbeta$ based on  $\*Y$ and $\*X$ is a surrogate $L_0$-penalized estimator defined as the limit of
the following iteratively reweighted ridge regression algorithm:
\begin{eqnarray}
\nonumber
{\bbeta}^{(k)}&=&\arg\min_{\bm \beta} \{ \|\*Y-\*X\bm\beta\|^2+\lambda_n  \sum_{j=1}^{p_n}\frac{\beta_j^2}
{\{{\bbeta}^{(k-1)}_j\}^2}\} \\
&=&\{\*X^\top\*X+\lambda_n \*D({\bbeta}^{(k-1)})\}^{-1}\*X^\top\*Y, \quad k=1,2,...
\end{eqnarray}
where ${\bm \beta}^{(0)}=\arg\min_{\bm \beta} \{ \|\*Y-\*X\bm\beta\|^2+\xi_n  \sum_{j=1}^{p_n} \beta_j^2\}
=(\*X^\top\*X+\xi_n\*I)^{-1}\*X^\top\*Y $ is an initial ridge estimator, $\xi_n>0$ and $ \lambda_n\ge 0$ are tuning penalization parameters, and
 for any $p_n$-dimensional vector $\btheta = (\theta_1, ..., \theta_{p_n})^T$, $\*D(\btheta)=\mbox{diag}(\frac{1}{\theta_1^2}, ..., \frac{1}{\theta_{p_n}^2})$. Note that
each reweighted $L_2$    {penalty} can be regarded as an adaptive surrogate $L_0$ penalty and the approximation of $L_0$ penalization improves with each iteration.
\citet{DAI2018334} showed that  the BAR estimator
$\hat{\bbeta} = \lim_{k\to\infty} {\bbeta}^{(k)}$ is selection consistent and possesses an oracle property: if the true model is sparse with some zero coefficients, then with probability tending to 1,
BAR estimates the true zero coefficients as zeros and estimates the
non-zero coefficients as well as the scenario when the true sub-model is known in advance.

\subsection{Broken adaptive ridge estimator for censored data (CBAR)}
For right censored data, the above BAR algorithm is obviously not applicable since one only observes $(T_i, \delta_i)$  instead of $Y_i$. To overcome the problem, we propose to adopt the  \citet{Leurgans1987} synthetic data approach for censored linear regression to variable selection by first transforming $(T_i, \delta_i)$ into a synthetic variable $Y_i^*$ and then applying the BAR method to the synthetic data
variable $Y_i^*$. Specifically,
the \citet{Leurgans1987} synthetic data $Y_i^*$ is defined as
\begin{equation}
\label{model:Leurgans}
Y_i^*=\int_{-\infty }^{T^n} {\left(\frac{I(T_i\geq s)}{1-\hat{H}(s)}- I(s\textless{0})\right)}ds,
\end{equation}
where $T^n =\max\{T_1, ..., T_n\}$ and $\hat{H}$ is the Kaplan-Meier estimator of $H$.
To apply the BAR method to synthetic data  $Y_i^*$, let $\*Y^*= (Y^*_1,...,Y^*_n)^\top$ and  define  an initial ridge estimator
\begin{eqnarray}
\hat{\bm \beta}^{(0)}
=(\*X^\top\*X+\xi_n\*I)^{-1}\*X^\top\*Y^*,
\label{ridge}
\end{eqnarray}
and then, for $k\ge 1$,
\begin{equation}
\hat{\bbeta}^{(k)}=g(\hat{\bbeta}^{(k-1)}),
\end{equation}
where
\begin{eqnarray}
\label{eq:one}
g(\tilde{\bm\beta}) =\arg\min_{\bm \beta} \{ \|\*Y^*-\*X\bm\beta\|^2+\lambda_n  \sum_{j=1}^{p_n}\frac{\beta_j^2}{\tilde{\beta}_j^2}\}
=\{\*X^\top\*X+\lambda_n \*D(\tilde{\bm\beta})\}^{-1}\*X^\top\*Y^*.
\label{eq:onea}
\end{eqnarray}
Finally, the CBAR estimator is defined as
\begin{equation}
\label{BJ BAR}
\hat{\bbeta}^*=\lim_{k\rightarrow\infty} \hat{\bbeta}^{(k)}.
\end{equation}
In the next section, we give conditions under which the CBAR estimator $\hat{\bm\beta}^*$ is selection consistent and has
an oracle property for estimation of the nonzero component $ \bbeta_{01}$ of $ \bbeta$.

\subsection{Large sample properties of CBAR}
Similar to \citet{Zhou1992Asymptotic}, define
$F_i{(t)}=P\{Y_i \geq t\}$, $ G_{i}(t)=P\{T_i \geq t\}=F_i{(t)}(1-H(t))$, $ K(t)= -\int_{0}^{t}{\frac{1}{\lim{(1/n)}\sum{F_i}}}\frac{dG}{G^2}$ and denote
$$\Lambda_i^{+}{(t)}=-\int_{0}^{t}{\frac{dG_i{(s)}}{G_i{(s^-)}}},\ \Lambda_i^D{(t)}=-\int_{0}^{t}{\frac{dF_i{(s)}}{F_i{(s^-)}}},\ \Lambda^C{(t)}=\int_{0}^{t}{\frac{dH{(s)}}{1-H{(s^-)}}}.$$
Then,
$$M_{i}^{+}(t)=I_{[T_i\leq t]}-\int_{0}^{t}{I_{[T_i\geq s]}d\Lambda_i^{+}{(s)}},$$
$$M_{i}^{D}(t)=I_{[T_i\leq t; \delta_i=1]}-\int_{0}^{t}{I_{[T_i\geq s]}d\Lambda_i^{D}{(s)}},$$
$$M_{i}^{C}(t)=I_{[T_i\leq t; \delta_i=0]}-\int_{0}^{t}{I_{[T_i\geq s]}d\Lambda_i^{C}{(s)}}$$
are square-integrable martingales and satisfies
$M_{i}^{+}=M_{i}^{D}+M_{i}^{C}$
\citep{Zhou1992Asymptotic}.
Let $\*\Omega(\tau)=( \sigma_{kl}(\tau))$ be defined by
\begin{eqnarray}
\label{sigmma_leurgans}
  \nonumber \sigma_{kl}(\tau)&=&\lim n \sum_{i=1}^n {\omega_{ki}\omega_{li}}\int_{0}^{\tau}{\left [\int_{t}^{\tau}{F_i}{ds} \right ]^2{\frac{d\Lambda_i^D{(t)}}{G_i}}}\\
  &+& \lim n\sum_{i=1}^n{\int_{0}^{\tau}{\prod_{c_{i}=\omega_{ki},\omega_{li}}{\left [ \frac{\sum{c_{j}}\int_{t}^{\tau}{F_j{ds}}}{(1-H)\sum{F_j}}-\frac{c_{k}{\int_{t}^{\tau}{F_i{ds}}}}{G_i} \right ]}G_i}{d\Lambda^C}},
\end{eqnarray}
where $\omega_{ji}=((\*X^\top \*X)^{-1}\*X^\top )_{ji}$.
Let {$\bm\omega_{i}$ denote the $i$th column of the matrix $(\*X^\top \*X)^{-1}\*X^\top$,}
$\*X_1$ denote the first $q_n$ columns of $\*X$,  $ \*\Sigma_{n}=n^{-1} \*X^\top \*X$ and $\*\Sigma_{n1}=n^{-1} \*X_1^\top\*X_1$. Write
$ \hat{\bm\beta}^*=( \hat{\bm\beta}_{1}^{*^\top }, \hat{\bm\beta}_{2}^{*^\top })^\top$, where $\hat{\bm\beta}_{1}^{*}$ is a ${q}\times 1$ vector estimator of $ \bbeta_{01}$ and $\hat{\bm\beta}_{2}^{*}$  is a $(p_n-{q})\times 1$ vector estimator of $ \bbeta_{02}$.

\par The following conditions are needed for our theoretical derivations.
\begin{itemize}
\itemsep0em
\itemsep0em
  \item[(C1)] 
 $\sup_{t} E(\varepsilon_i-t |\varepsilon_i> t )<\infty$, and  {\color{black}for any $p_n$-vector $ \*b_n$ satisfying {$\| \*b_n\|\leq 1$}, $\*b_n^\top\*\Omega(\tau)\*b_n$ }is finite for $\tau\in[K,\infty]$ and {$\*b_n^\top\*\Omega(\tau)\*b_n \to \*b_n^\top\*\Omega(\infty)\*b_n$ } as $\tau\to \infty$.
  \item[(C2)]
   $\sup_n \int_{0}^{\infty}\sum_{i=1}^n {(\*b_n^\top\bm\omega_{i} )}^2 \sum_{i=1}^n  {F_{i}^2{dK(t)}}< \infty$ for {any $p_n$-vector $ \*b_n$ satisfying {$\| \*b_n\|\leq 1$}}.
   {$X_i$ are bounded, and for some constants $C^*>0$ and $S<1$, $C^*F_i(t)^S\leq 1-H(t)$.}
   \item[(C3)]
   $\int_{0}^{\infty}\{{\frac{n\sum ({\*b_n^\top\bm\omega_{i} })^{2}F_{i}}{1-H(s)}\}^{\frac{1}{2}}ds}\leq M < \infty$ and $\int_{0}^{\infty}{K^{1/2}(t)|\sum {\*b_n^\top\bm\omega_{i} }F_{i}|dt} <\infty$ for {any $p_n$-vector $ \*b_n$ satisfying {$\| \*b_n\|\leq 1$}}.
  \item[(C4)]
  There exists a constant $\tilde{C}>1$ such that $0<1/ \tilde{C}<\lambda_{\min}( \*\Sigma_n)\leq \lambda_{\max}(\*\Sigma_n)< \tilde{C} <\infty$ for every integer $n$.
  \item[(C5)]
  Let
  $a_{{{0}}}= \min_{1\leq j\leq {q}}|\beta_{0j}|$ and $a_{{{1}}}= \max_{1\leq j\leq {q}}|\beta_{0j}|$.
  As $n\to\infty$, ${{{p_n}}}/{\sqrt{n}} \to  0$, { ${\xi_n}/\sqrt{n} \to  0$} and {$\lambda_n/\sqrt{n} \to  0$}
\end{itemize}
Conditions (C1)-(C3) are regularity conditions required to establish the asymptotic properties of the unpenalized synthetic data least squares estimator under diverging dimension.
Conditions (C4) and (C5) are additional conditions needed to derive the selection consistency and oracle property of the synthetic data BAR estimator of this paper as stated in Theorem 1 below
\begin{theorem}[Oracle property]
\label{theorem:CBAR}
Assume conditions (C1)-(C5) hold.  For any $q$-dimensional vector $ {\*c}$ satisfying $\| {\*c}\|\leq 1$, define ${z^2= \*c^\top \*\Omega_{1} \*c}$, where $\*\Omega_{1}$ is the first ${q \times q}$ sub-matrix of $\*\Omega(\infty)$.
Define {$f( \bm\alpha)= \{ \*X_{1} ^\top   \*X_1+\lambda_n \*D_1( \bm\alpha)\}^{-1} \*X_1^\top \*Y^*$}, where $ \*D_1( \bm\alpha) = {\rm diag}(\alpha_1^{-2},\ldots,\alpha_{{q}}^{-2})$. Then, with probability tending to 1,
\begin{enumerate}
\item[(i)]
${{\widehat{\bbeta }}^{*}}={{(\widehat{\bbeta }_{1}^{{{*}^{\top }}},\widehat{\bbeta }_{2}^{{{*}^{\top }}})}^{\top}}$ exists and is unique, with $\hat{\bm\beta}_{2}^{*}=0$ and $\hat{ \bbeta}_1^*$ being the unique fixed point of $f(\bm\alpha)$;
\item[(ii)]
 $\sqrt{n}\, {z^{-1} \*c}^\top(\hat{ \bbeta}_{1}^{*}- \bbeta_{01})
\rightarrow_D N(0,1).$
\end{enumerate}
\end{theorem}
Part (i)  of the above theorem guarantees  that the CBAR estimator is consistent for variable selection. Part (ii) states that the asymptotic distribution of the nonzero component of the CBAR estimator is the same as the one when the true model is known in advance. The proof of Theorem \ref{theorem:CBAR} is deferred to the Appendix.
{
\subsection{Grouping effect}\label{group_simulation}
When the true model has a group structure, it would be desirable for a variable selection method to either retain or drop all variables that are clustered within the same group.
 Below we establish that the CBAR estimator possesses a grouping property in the sense that  highly correlated covariates tend to be grouped together with similar coefficients.
\begin{theorem}
	\label{theorem:2}
	{\it
Assume that the columns of matrix $\*X$ are standardized and $\*Y^*$ is centered.
Let $\hat{ \bbeta}^*$ be the CBAR estimator and {$\hat{\beta}_i^*\hat{\beta}^*_j > 0$}, then, with probability tending to $1$,
\begin{equation}
\label{groupEquation1}
		|\hat{\beta}_i^{*-1} - \hat{\beta}_j^{*-1}|\leq \frac{1}{\lambda_n}\, \| \*Y^*\|\sqrt{2(1- r_{ij})},
\end{equation}	
where ${r_{ij}= \*x_i^\top \*x_j}$ is the sample correlation of $\*x_i$ and $ \*x_j$.}
\end{theorem}

The above result implies that the estimated coefficients
of two highly positively-correlated variables will be similar in magnitude.
The proof of Theorem \ref{theorem:2} is given in the Appendix. Similarly, it can be shown
that the estimated coefficients
of two highly negatively-correlated variables will also be similar in magnitude.
}

{
\subsection{Ultrahigh dimensional covariates}\label{highdim}
 Theorem 1  is established under a sufficient condition that $p_n < n$. In many applications,  $p_n$ can be much larger than the sample size $n$. For high dimensional problems, a common strategy is to proceed a variable selection method with a sure screening dimension reduction step  \citep{Fan2008Sure,Zhu2011Model,Cui2015Model}. This strategy also applies to the semiparametric AFT model with right censored data. For  example, one can first apply the sure joint screening method BJASS of \citet{YiLiu} to obtain a lower dimensional model and then apply the CBAR method to the reduced model. We refer to the resulting two-step estimator $\hat{ \bbeta}^{*}$ as the BJASS-CBAR estimator.

Below we give some additional sufficient conditions under which the BJASS-CBAR estimator $\hat{ \bbeta}^{*}$ has an oracle property.
}
{
\begin{itemize}
 \item [(D1)] $\log(p)=O(n^d)$ for some $0\leq d < 1$.

 \item [(D2)] $P(t\leq Y_i \leq C_i)\geq \tau_0>0$ for some positive constant $\tau_0$ and any $t\in [0,\varsigma]$, where $\varsigma$ denotes the maximum follow up time. Furthermore, $\sup\{t:P(Y>t)>0\}\geq \sup\{t:P(C>t)>0\} $. $H(t)$ has uniformly bounded first derivative.
  \item [(D3)]  $\min_{j\in s^*}|\beta_j^*|\geq \omega_1n^{-\tau_1}$  and $q<k\leq \omega_2n^{\tau_2}$ for some positive constants $\omega_1,~\omega_2$ and nonnegative constants $\tau_1,~\tau_2$ satisfying $\tau_1+\tau_2<1/3$, where $k$ is the size of the screened model from BJASS.

    \item [(D4)] For sufficiently large $n$, $\lambda_{\min}(n^{-1} \*X_s ^\top \*X_s)\geq c_1$ for some constant $c_1>0$ and all $s\in S_{+}^{2k}$, where $\lambda_{\min}(\cdot)$ denotes the smallest eigenvalue of a matrix, and $S_{+}^{k}=\{s:s^*\subset s; \|s\|_0\leq k\}$ denotes the collection of the over-fitted models of cardinality $k$ or smaller.
 \item [(D5)] Let $\sigma_i^2=\int\int[\frac{G_i(s\vee t)}{(1-H(s))(1-H(t))}-F_i(s)I(t<0)-F_i(t)I(s<0)+I(s<0)I(t<0)]dsdt-E^2(Y_i)$. There exist positive constants $c_2$, $c_3$, $c_4$, $\sigma$ such that $|X_{ij}|\leq c_2$, $|X_i^\top\bbeta^*|\leq c_3$, $|\sigma_i|\leq \sigma$ and for sufficiently large $n$,
     \begin{eqnarray*}
     \max_{1\leq j\leq p}\max_{1\leq i \leq n}\left\{\frac{X_{ij}^2}{\sum_{i=1}^nX_{ij}^2\sigma_i^2}\right\}\leq c_4n^{-1}.
     \end{eqnarray*}
 \item [(D6)] There are positive constants $K_1$, $K_2$ and $\tau_3$ such that
     \begin{eqnarray*}
     P(|\epsilon|\geq M)\leq K_1\exp(-K_2M^{\tau_3}),
     \end{eqnarray*}
  for any $M=O(n^{\tau})>0$, where $\tau\geq 0$, $\tau_1+\tau_2+\tau<(1-d)/2$, $\tau_2+d-\tau\tau_3<0$, and $2\tau_2+2\tau+d<1/3$.
\end{itemize}

\begin{theorem}[Oracle property of the BJASS-CBAR estimator]\label{BJASS_CBAR}
Assume that conditions (D1)-(D6) hold and that the assumptions of Theorem 1 hold for the BJASS reduced model. Then, with probability tending to 1,
\begin{enumerate}
\item[(i)] $\hat{\bm\beta}_{2}^{*}=0$;
\item[(ii)]
 $\hat{ \bbeta}_{1}^{*}$ performs as well as the oracle estimator for the true model $\mathcal{M}_{*}=\{1\le j\le q\}$ in the sense of part (ii) of Theorem 1.
\end{enumerate}
\end{theorem}

The above result is a direct consequence of  Theorems 4 of \citet{YiLiu} and the oracle property of CBAR stated in Theorem 1.
In Section \ref{sect:Simulation2}, we present a simulation study to illustrate the advantages of BJASS-BAR
in comparison with some other penalization methods under a high dimensional setting.
}

\section{Simulations}\label{Numerical}
We present some simulations to illustrate
the effectiveness of the proposed CBAR estimator for variable selection, prediction, parameter estimation 
in comparison with some popular penalization methods including
Lasso \citep{tib1996}, adaptive Lasso \citep{zou2006}, SCAD \citep{fan2001} and MCP \citep{zhang2010}),
 in the context of the \citet{Leurgans1987} synthetic data framework. We use the  \textsf{R} package \texttt{glmnet} \citep{friedman2010} for Lasso and adaptive Lasso and \textsf{R} package
\texttt{ncvreg} \citep{breheny2011} for SCAD and MCP, performed on the \citet{Leurgans1987} synthetic data
outcome. Five-fold cross-validation (CV) is used to select tuning parameters for all methods.
For CBAR, we all use 10 equally log-spaced grid points on $[a, b]$ for the paths of $\lambda_n$ and $\xi_n$ where $a=1e^{-4}$ and $b=\mbox{max}\left\{\frac{(\*x^{\T}_j\*y)^2}{4\*x^{\T}_j\*x_j}\right\}^p_{j=1}$.

\subsection{Simulation 1:  $p_n < n$ }\label{sect:Simulations}

We consider the following two model settings
similar to  \citep{Tibshirani1997, Fan2002, Cai2009Regularized}:
\leftmargini=2.5cm
\begin{itemize}
\item[Model 1:] $Y_i= \*x_i^\top \bm\beta_0+\varepsilon_i$, where
the covariate vector $ \*x_i$ is generated from a multivariate normal distribution with mean $0$ and variance-covariance matrix $\*\Sigma=(\rho^{|i-j|})$, and the error $\varepsilon_i$ has the standard normal distribution and is independent of the covariates.\\
The true parameter value is $\bm\beta_0 = (3,-2, 0, 0, 6, 0,\ldots, 0)^\top$.
\item[Model 2:]
 The same as Model 1 except that\\
$\bm\beta_0 = (3,-2, 6, 0.3, -0.2, 0.6, 0, \ldots, 0)^\top $.
\end{itemize}
Note that Model 1 contains  strong signals, whereas Model 2 includes both strong and weak signals.
The censoring variable $C_i$ is generated from the normal distribution $ N(c, 2)$, where $c$ is chosen to yield a desired level of censoring rate.

The variable selection performance is assessed using five measures:
the mean number of misclassified non-zeros and zeros (MisC), mean of false non-zeros (FP), mean of false zeros (FN),
probability that the selected model is identical to the true model (TM), and a similarity measure (SM)
between the selected set $\hat S $ and the true active set $ |S|_{0}$:
$SM= \frac{ |\hat{S}\cap S|_{0}}{\sqrt{|\hat{S}|_{0} |S|_{0}}},$
where $|.|_{0}$ denotes model size.  The prediction performance is measured by
the mean squared prediction error (MSPE) from the five-fold CV.  The parameter
estimation performance is measured by the mean of the absolute bias of the parameter estimator (MAB). We have run extensive simulations for a variety of settings by varying $n$, $p$, $\rho$ and the censoring rate, with
{1,000} Monte Carlo replications for each setting. Part of the findings are presented in Table \ref{table:1}.

\begin{center}
[Insert Table \ref{table:1} approximately here]
\end{center}

\begin{table}
\centering
\caption{\small {Comparison of CBAR with Lasso, SCAD, MCP, and Adaptive Lasso (ALasso) when coupled with the \citet{Leurgans1987} synthetic data procedure based on {{1,000}} Monte-Carlo replications. Data settings: $n=100$, $p \in \{10, 50, 80, 90 \}$. (MisC = mean number of misclassified non-zeros and zeros; FP = mean of false positives (non-zeros); FN = mean of false negatives (zeros); TM = probability that the selected model is exactly the true model; SM = similarity measure; MSPE = mean squared prediction error from five-fold CV or five-jointly CV and MAB = mean of the absolute bias of the parameter estimator.)}}
\bigskip
\scalebox{0.8}{
\begingroup
\setlength{\tabcolsep}{5pt} 
\renewcommand{\arraystretch}{0.9} 
\begin{tabular}{clllllllll}
\hline\hline
 Model& p & Method & MisC & FP & FN & TM & SM & MSPE & MAB \\
\hline
  1& 10  & CBAR & $\bm{0.60}$ & $\bm{0.60}$ & 0 & $\bm{74\%}$ & $\bm{0.94}$ &$\bm{8.86}$  & 1.50 \\
  & & Lasso & 3.05 & 3.05 & 0 & 6.6$\%$ & 0.73 & 9.28 & 2.29\\
  & & SCAD & 1.11 & 1.11 & 0 & 46.2$\%$ &0.89 & 9.03 & 1.47 \\
  & & MCP & 0.76 & 0.76 & 0 & 63.6$\%$ & 0.92 & 9.01 & 1.45 \\
  & & Alasso & 1.12 &1.12 & 0 & 49.2$\%$ & 0.89  & 8.91 & 1.68 \\
  \cline{2-10}
  & 50  & CBAR & $\bm{0.73}$ &  $\bm{0.71}$ & 0.02 & $\bm{74.80\%}$ & $\bm{0.94}$ & $\bm{8.9}$  & 1.69\\
  & & Lasso & 7.33 & 7.33 & 0 & 1.7$\%$ &0.58 & 9.77 & 3.36 \\
  & & SCAD & 2.96 & 2.96 & 0 & 21.3$\%$ & 0.76  & 9.09 & 1.72\\
  & & MCP & 1.24 & 1.23 & 0.01 & 47.7$\%$ & 0.88 & 9.03 & 1.56 \\
  & & Alasso & 6.09 & 6.09 & 0 & 15.2$\%$ & 0.67  & 8.69& 3.04 \\
  \cline{2-10}
 & 80  & CBAR & $\bm{0.86}$ & $\bm{0.84}$  & 0.02 & $\bm{72.3\%}$ &$\bm{0.93}$ & 8.84 & 1.81 \\
  & & Lasso & 9.40 & 9.40 & 0 & 1.30$\%$ & 0.54 & 10.06 & 3.79\\
  & & SCAD & 3.90 & 3.90 & 0 &15.2$\%$ &0.72  & 9.33 & 1.89\\
  & & MCP & 1.41 & 1.40 & 0.01 & 45.6$\%$ & 0.87 & 9.26 & 1.65 \\
  & & Alasso & 11.09 & 11.08  & 0.01 & 11.8$\%$ & 0.59 & $\bm{8.74}$  & 4.51\\
  \cline{2-10}
  & 90  & CBAR &  $\bm{0.94}$ & $\bm{0.92}$ & 0.02 & $\bm{69.7\%}$ &$\bm{0.93}$ & $\bm{8.93}$   & 1.88 \\
  & & Lasso & 9.36 & 9.36 & 0 & 1.4$\%$ &0.54 & 10.05 & 3.82 \\
  & & SCAD & 4.09 & 4.09 & 0 & 13.5$\%$ & 0.71 & 9.27 & 1.91\\
  & & MCP & 1.44 & 1.43 & 0.01 & 43.2$\%$  &0.87 & 9.20 & 1.64 \\
  & & Alasso & 4.29 & 4.27 & 0.02 & 10.5$\%$ &0.70   &9.13 & 2.69 \\
  \hline
   2 & 10  & CBAR & 2.61  & $\bm{0.65}$ & 1.96 & 0.9$\%$ & 0.77 & 9.36 & 2.36  \\
  & & Lasso & 3.00 & 2.14 & $\bm{0.86 }$ & 2.2$\%$ &0.78 & 9.50 & 2.74 \\
  & & SCAD & 2.64 & 1.10 & 1.54 & 2.3$\%$ &  0.78  & 9.35 & 2.35\\
  & & MCP & 2.64 & 0.86 & 1.78 & 1.9$\%$ & 0.77 & 9.34 & 2.36\\
  & & Alasso & $\bm{2.50}$ & 0.92 & 1.58 & $\bm{3.1}\%$ & $\bm{0.79}$ & $\bm{9.14}$  & 2.37\\
   \cline{2-10}
 & 50  & CBAR & $\bm{3.65}$ & $\bm{1.03}$ & 2.62 & 0.1$\%$ & $\bm{0.69}$  & 9.41 & 2.92\\
  & & Lasso & 9.75 & 7.99 & $\bm{1.76}$ & 0$\%$ & 0.52 & 10.40 & 4.37 \\
  & & SCAD & 5.57 & 3.40 & 2.17 & 0$\%$ & 0.61 & 9.84 & 2.77 \\
  & & MCP & 3.92 & 1.46 & 2.46 & 0$\%$ &0.67  & 9.80 & 2.64  \\
  & & Alasso & 9.18 & 7.22 & 1.96 & 0.1$\%$  &0.55 & $\bm{9.21}$  & 4.27 \\
  \cline{2-10}
  & 80  & CBAR & $\bm{3.89}$ & $\bm{1.19}$ & 2.70 & 0$\%$ & $\bm{0.68}$ & 9.11 & 3.02 \\
  & & Lasso & 11.61 & 9.69 & $\bm{1.92}$  & 0$\%$ & 0.48 & 10.39 & 4.70 \\
  & & SCAD & 6.48 & 4.21 & 2.27 & 0$\%$ & 0.57 & 9.66 & 2.86 \\
  & & MCP & 4.06 &1.49 & 2.57 & 0$\%$ & 0.66  & 9.60 & 2.64 \\
  & & Alasso & 13.82 & 11.78 & 2.04 & 0$\%$ & 0.48 & $\bm{8.99}$  & 5.55\\
   \cline{2-10}
  & 90  & CBAR &$\bm{3.85}$ & $\bm{1.16}$ & 2.69 &0$\%$ & $\bm{0.68}$& 9.31 & 3.05 \\
  & & Lasso & 12.44 & 10.47 & $\bm{1.97}$  & 0$\%$ & 0.46 & 10.20 & 4.86 \\
  & & SCAD & 6.92 & 4.67 & 2.25 & 0$\%$ &0.56  & 9.40 & 2.92 \\
  & & MCP & 4.24 & 1.68 & 2.56 & 0$\%$ & 0.65  & 9.36& 2.68 \\
  & & Alasso & 7.00 & 4.50 & 2.50 & 0$\%$ & 0.54 & $\bm{9.27}$ & 3.73 \\
  \hline
\end{tabular}
\endgroup
}
\label{table:1}
\end{table}

It is seen from Table \ref{table:1} that CBAR stands out as the top or top two performers with respect to almost all variable selection performance measures (MisC, FP, TM and SM).
{In particular, CBAR generally yields a more sparse and accurate model with the largest TM and SM, and much lower MisC and FP}. Also, using fewer active features, CBAR achieves comparable  prediction accuracy as other methods that use more features. For estimation, CBAR, SCAD and MCP are comparable with similar bias (MAB), whereas Lasso and {Adaptive lasso} can be substantially worse. 

\subsection{Simulation 2:  $p_n >> n$ }\label{sect:Simulation2}
In this simulation, we consider the same models as in Simulation 1, except in a high dimensional setting with $ n=200$, $p=1000$.
We again compared the same five penalization methods, with each method proceeded with the sure joint screening method BJASS of \citet{YiLiu} with $k=2log(n)*n^{(1/4)}$ for the semi-parametric AFT model to yield a two-step sparse estimator. We denote these methods by BJASS-CBAR, BJASS-Lasso, BJASS-SCAD, BJASS-MCP and BJASS-ALasso.
The censoring rate is 0.2. The results are summarized in Table \ref{table:3}.
\begin{center}
[Insert Table \ref{table:3} approximately here]
\end{center}

\begin{table}[h!]
\centering
	\caption{\small
Comparison of BJASS-CBAR with CBAR with BJASS-Lasso, BJASS-SCAD, BJASS-MCP, and BJASS-ALasso when coupled with the \citet{Leurgans1987} synthetic data procedure in a high-dimensional setting: $ n=200$, $p=1000$. (MisC= mean number of misclassified non-zeros and zeros; FP = mean of false positives (non-zeros); FN = mean of false negatives (zeros); TM = probability that the selected model is exactly the true model; SM = similarity measures; MSPE = mean squared prediction error from five-fold CV or five-jointly CV and MAB = mean of the absolute bias of the parameter estimator.)}
\bigskip	
	\scalebox{0.8}{
	\begingroup
	\setlength{\tabcolsep}{4pt} 
	\renewcommand{\arraystretch}{1} 
	\begin{tabular}{clllllllr}
       \hline
	 Model & Method & MisC & FP & FN & TM & SM & MAB & MSPE\\
			\hline
	1 & {BJASS-CBAR} & $\bm{2.24 }$ & $\bm{2.15 }$ & 0.09 &$\bm{63\%}$ & $\bm{0.93}$ & 2.87 & 10.40  \\
      & {BJASS-Lasso} & 12.61  & 12.55  & 0.06  & 0$\%$ & 0.63   & 4.79  & 10.87 \\
      & {BJASS-SCAD}  & 4.23  & 4.14  & 0.09  & 20$\%$  & 0.82  & 2.79 & 10.46 \\
       & {BJASS-MCP}  & 2.82  & 2.73  & 0.09  & 43$\%$  & 0.88 & $\bm{2.69}$   & 10.45  \\

      & {BJASS-ALasso} & 8.08  & 8.00  & 0.08  & 12$\%$  & 0.73  & 4.05  & 10.35 \\
     \hline
    2 & {BJASS-CBAR}  & $\bm{6.15 }$  & $\bm{3.15}$   & 3  & $\bm{41\%}$ & $\bm{0.69}$ & 2.51  & 12.17  \\
        & {BJASS-Lasso} & 17.14 & 14.14 & 3  & 0$\%$  & 0.49  & 4.49  & 12.64\\
        & {BJASS-SCAD}  & 8.68  & 5.68  & 3  & 7$\%$  & 0.62 & 2.09  & 12.39 \\
        & {BJASS-MCP}   & 6.38  & 3.38 & 3 & 26$\%$  & 0.68 & $\bm{1.96}$ & 12.38 \\
       & {BJASS-ALasso} & 12.78  & 9.78  & 3  & 3$\%$ & 0.54 & 3.75  & 11.91  \\
          \hline
		\end{tabular}		
		\endgroup
	}\\
		\label{table:3}
\end{table}
{It is observed from  Table \ref{table:3} that although most penalization methods had comparable
performance in terms of estimation bias (MAB) and prediction error (MSPE), BJASS-CBAR outperformed the other methods in the variable selection domain with the lowest MisC, FP and the largest TM and SM, which are consistent with the simulation results for the low-dimension $p_n<n$ settings in Simulation 1.
}

\section{Real data examples}
We illustrate the CBAR method on two real  datasets with high dimensional covariates.

\subsection{Diffuse large-B-cell lymphoma data}
The diffuse large-B-cell lymphoma (DLBCL) data includes $n=240$ patients and $p=7399$ gene features, which was downloaded from \url{http://statweb.stanford.edu/~tibs/superpc/staudt.html}. We first apply the BJASS sure joint screening method  of \citet{YiLiu} to reduce data dimension to  $k=2log(n) n^{\frac{1}{4}}=43$ and then apply CBAR and four other popular penalization methods. The results are summarized in Table \ref{DLBCL}.

\begin{center}
[Insert Table \ref{DLBCL} approximately here]
\end{center}

\begin{table}[h!]
	\centering
	\caption{Estimated coefficients of {BAJSS-CBAR, BAJSS-Lasso, BAJSS-SCAD, BAJSS-MCP and BAJSS-Alasso for the DLBCL data}.}	
    \bigskip
    \scalebox{0.85}{
	\setlength{\tabcolsep}{2pt} 
	\renewcommand{\arraystretch}{1} 
	\begin{tabular}{lccccr}
	Parameter & BAJSS-CBAR & BAJSS-Lasso &BAJSS-SCAD & BAJSS-MCP & BAJSS-Alasso\\
			\hline			
$1456$ &-0.0591  &$-0.394$&$-0.609$&-0.630 &$-0.513$\\
$1819$ & &$-0.069$ & & &\\
$1863$ & &-0.006 & & &\\
$2603$ & &$-0.025$ & & &\\
$2672$ & &$-0.062$ & & &\\
$3236$ &-0.480 &-0.348 &-0.394 &-0.426 &-0.399\\
$5775$ &-0.261 &$-0.143$ &-0.133 &-0.131 &$-0.111$\\
$6566$ & &-0.088 & -0.061& -0.004&\\
			\hline
	Tuning parameters & \tabincell{l}{\  $\xi_n=43$\\ $\lambda_n=5.721$}&$\lambda=0.197$&\tabincell{l}{\ $\gamma=3.7$,\\ $\lambda=0.211$}& $\lambda=0.260$ &\tabincell{c}{\ $\gamma=3.598$,\\ $\lambda= 2.058$}\\
			\hline
	{Number of selected}  &3 & 8& 4 & 4& 3\\
            \hline
	{CV error}  &6.399 & 6.731 & 6.496  & 6.515& 6.472\\
			\hline
	\end{tabular}} \label{DLBCL}
\end{table}
It is seen that BJASS-CBAR is among the most sparse model and has the smallest  CV error, which is consistent with the findings in the simulation studies.
{
\subsection{Glioblastoma multiforme data} \label{GBMdata}

The glioblastoma multiforme (GBM) methylation data was downloaded from the TCGA program (https://www.cancer.gov/tcga) using TCGA-Assembler 2 (TA2).  The initial data consists of
577 patients and 20,156 GBM methylation variables.  After removing missing data, the complete case data includes $n=136$ patients and $p=20,037$ methylation variables. Applying the method described in Section
2.5, we first performed  sure joint screening using the BJASS method of \citet{YiLiu} reduce data dimension to $k=2 log(n) n^{\frac{1}{4}} =34$ before applying the CBAR penalization method and four other penalization methods (Lasso, SCAD, MCP and Alasso). The final variable selection results are summarized in the  Table \ref{GBM_M_1}.
\begin{center}
[Insert Table \ref{GBM_M_1} approximately here]
\end{center}
\begin{table}[h!]
	\centering
	\caption{{Estimated coefficients of BJASS-CBAR, BJASS-Lasso, BJASS-SCAD, BJASS-MCP and BJASS-Alasso for the TCGA GBM methylation data}}	
    \bigskip
    \scalebox{0.85}{
	\setlength{\tabcolsep}{2pt} 
	\renewcommand{\arraystretch}{1} 
	\begin{tabular}{lccccr}
	Variables & BJASS-CBAR & BJASS-Lasso &BJASS-SCAD & BJASS-MCP & BJASS-Alasso\\
			\hline	
BCL2L10&	&	0.051& 	0.038& 	&	0.038\\
CDCP2&	-0.272& 	-0.077& 	-0.057 &	&	-0.068\\
HES5&	&	-0.139 &	-0.153& 	-0.265 &	-0.162\\
HLA.E	&&	0.104 &	0.117 &	0.167 &	0.098\\
HRH3&	&	0.021 &	&	&	\\
IRX6&	&	0.014 &	&	&	\\
KIF5C&	&	0.004& 	&	&	\\
NIPSNAP3B&	&	0.034 &	& 	&	0.017\\
NPM2	&0.230 &	0.087& 	0.065& 	0.089 &	0.078\\
OXGR1&	&	0.059 &	0.066& 	&	0.045\\
SLC12A5	&0.282 &	0.144& 	0.104& 	0.072 &	0.167\\
SMIM11A&	0.417 &	0.349 &	0.469& 	0.507& 	0.418\\

			\hline
	Tuning parameters & \tabincell{l}{\  $\xi_n=19$\\ $\lambda_n=1.642$}&$\lambda=0.122$&\tabincell{l}{\ $\gamma=3.7$,\\ $\lambda=0.154$}& $\lambda=0.190$ & $\lambda= 0.625$\\
			\hline
	{Number of selected}  &4 & 12& 9 & 5& 9\\
            \hline
	{CV error}  &3.793	&3.832	&3.804&	3.835&	3.620\\
			\hline
	\end{tabular}} \label{GBM_M_1}
\end{table}

It is seen from Table \ref{GBM_M_1} that our BJASS-CBAR selected the sparsest model with 4 variables while achieving a comparable CV error as compared to the other four methods, which is consistent with our findings in simulation studies. It is interesting to note that  the four features selected by BJASS-CBAR have also been selected by three other methods. Among the four selected features, NPM2 and IRX6  have been previously discussed in the literature to possibly  play critical roles with human diseases \citep{Eirin2006Long,Box2016Nucleophosmin,Daphna2019Germline,MUMMENHOFF2001193}.
}
\section{Discussion}\label{section5}
We have rigorously extended the broken adaptive ridge (BAR) penalization method for simultaneous variable selection and parameter estimation to the semiparametric AFT model with right-censored data by coupling BAR penalization with the \citet{Leurgans1987} synthetic data.  We have established  that the resulting CBAR estimator is asymptotically consistency for variable selection, has an oracle estimation property, {and enjoys a grouping property for highly correlated covariates}.  We consider both low and high dimensional covariate settings. Our empirical studies demonstrate that CBAR generally produces a more sparse and  accurate model as compared to some popular $L_1$-based penalization methods, which corroborates previous findings in the literature for uncensored data.

We note that coupling the BAR method with the \citet{Leurgans1987} synthetic variable is only one of several possible ways of extending the BAR method to right censored linear model for simultaneous variable selection and parameter estimation.
For example, one may couple the BAR method with the \citet{Koul1981Regression} synthetic data method, the \citet{stute1993consistent} weighted least squares method, or  the   \citet{Buckley1979Linear} iterative imputation method. Our limited numerical studies (not reported here) indicate that using \citet{Koul1981Regression} synthetic data is generally inferior to
using \citet{Leurgans1987} synthetic variable, whereas  iteratively performing BAR using the \citet{Buckley1979Linear} imputation
may sometimes improve the performance of the CBAR method based on the \citet{Leurgans1987} synthetic variable. However, asymptotic properties  of each of these distinct approaches require different theoretical developments. Thorough investigations and comparisons of  these alternative approaches are needed in  future research.

{Lastly, missing data often occurs in real world applications. Although there is a vast amount literature on missing data problems, little has been done to deal with missing data in the context of variable selection for survival data. Further research in this domain is waranteed.
}
{
\section{Acknowledgement}
The Glioblastoma multiforme data used in Section \ref{GBMdata} are generated by the TCGA Research Network: https://www.cancer.gov/tcga.
 }

\begin{thebibliography}{43}
\providecommand{\natexlab}[1]{#1}
\providecommand{\url}[1]{\texttt{#1}}
\expandafter\ifx\csname urlstyle\endcsname\relax
  \providecommand{\doi}[1]{doi: #1}\else
  \providecommand{\doi}{doi: \begingroup \urlstyle{rm}\Url}\fi

\bibitem[MUM(2001)]{MUMMENHOFF2001193}
Expression of irx6 during mouse morphogenesis.
\newblock \emph{Mechanisms of Development}, 103\penalty0 (1):\penalty0 193 --
  195, 2001.
\newblock ISSN 0925-4773.

\bibitem[Box(2016)]{Box2016Nucleophosmin}
Nucleophosmin: from structure and function to disease development.
\newblock \emph{Bmc Molecular Biology}, 17\penalty0 (1), 2016.

\bibitem[Akaike(1974)]{akaike1974new}
H.~Akaike.
\newblock A new look at the statistical model identification.
\newblock \emph{IEEE Trans. Automat. Contr.}, 19:\penalty0 716--723., 1974.

\bibitem[Breheny and Huang(2011)]{breheny2011}
P.~Breheny and J.~Huang.
\newblock Coordinate descent algorithms for nonconvex penalized regression,
  with applications to biological feature selection.
\newblock \emph{The annals of applied statistics}, 5\penalty0 (1):\penalty0
  232--253., 2011.

\bibitem[Breiman(1996)]{breiman1996heuristics}
L.~Breiman.
\newblock Heuristics of instability and stabilization in model selection.
\newblock \emph{Ann. Statist.}, 24:\penalty0 2350--2383, 1996.

\bibitem[Buckley and James(1979)]{Buckley1979Linear}
J.~Buckley and I.~James.
\newblock Linear regression with censored data.
\newblock \emph{Biometrika}, 66\penalty0 (3):\penalty0 429--436, 1979.

\bibitem[Cai et~al.(2009)Cai, Huang, and Tian]{Cai2009Regularized}
T.~Cai, J.~Huang, and L.~Tian.
\newblock Regularized estimation for the accelerated failure time model.
\newblock \emph{Biometrics}, 65\penalty0 (2):\penalty0 394--404, 2009.

\bibitem[Chen and Chen(2008)]{chen2008extended}
J.~Chen and Z.~Chen.
\newblock Extended bayesian information criteria for model selection with large
  model spaces.
\newblock \emph{Biometrika}, 95:\penalty0 759--771, 2008.

\bibitem[Cox(1972)]{Cox1972}
B.~D.~R. Cox.
\newblock Regression models and life-tables.
\newblock \emph{Journal of the Royal Statistical Society: Series B
  (Methodological)}, 34\penalty0 (2):\penalty0 187--220., 1972.

\bibitem[Cui et~al.(2015)Cui, Li, and Zhong]{Cui2015Model}
H.~Cui, R.~Li, and W.~Zhong.
\newblock Model-free feature screening for ultrahigh dimensional discriminant
  analysis.
\newblock \emph{Journal of the American Statistical Association}, 110\penalty0
  (510):\penalty0 630--641, 2015.

\bibitem[Dai et~al.(2018{\natexlab{a}})Dai, Chen, and Li]{dai2018fusedBAR}
L.~Dai, K.~Chen, and G.~Li.
\newblock The broken adaptive ridge procedure and its applications.
\newblock \emph{Stat Sin}, 2018{\natexlab{a}}.
\newblock \doi{10.5705/ss.202018.0075}.

\bibitem[Dai et~al.(2018{\natexlab{b}})Dai, Chen, Sun, Liu, and Li]{DAI2018334}
L.~Dai, K.~Chen, Z.~Sun, Z.~Liu, and G.~Li.
\newblock Broken adaptive ridge regression and its asymptotic properties.
\newblock \emph{Journal of Multivariate Analysis}, 168:\penalty0 334--351,
  2018{\natexlab{b}}.

\bibitem[Datta et~al.(2007)Datta, Le-Rademacher, and
  Datta]{Datta2007Predicting}
S.~Datta, J.~Le-Rademacher, and S.~Datta.
\newblock Predicting patient survival from microarray data by accelerated
  failure time modeling using partial least squares and lasso.
\newblock \emph{Biometrics}, 63\penalty0 (1), 2007.

\bibitem[Eirin-Lopez and J.(2006)]{Eirin2006Long}
Eirin-Lopez and M.~J.
\newblock Long-term evolution and functional diversification in the members of
  the nucleophosmin/nucleoplasmin family of nuclear chaperones.
\newblock \emph{Genetics}, 173\penalty0 (4):\penalty0 1835--50, 2006.

\bibitem[Fan and Li(2001)]{fan2001}
J.~Fan and R.~Li.
\newblock Variable selection via nonconcave penalized likelihood and its oracle
  properties.
\newblock \emph{Journal of the American statistical Association}, 96\penalty0
  (456):\penalty0 1348--1360, 2001.

\bibitem[Fan and Li(2002)]{Fan2002}
J.~Fan and R.~Li.
\newblock Variable selection for cox's proportional hazards model and frailty
  model.
\newblock \emph{The Annals of Statistics}, 30\penalty0 (1):\penalty0 74--99,
  2002.

\bibitem[Fan and Lv(2008)]{Fan2008Sure}
J.~Fan and J.~Lv.
\newblock Sure independence screening for ultrahigh dimensional feature space.
\newblock \emph{Journal of the Royal Statistical Society}, 70\penalty0
  (5):\penalty0 849--911, 2008.

\bibitem[Foster and George(1994)]{foster1994risk}
D.~Foster and E.~George.
\newblock The risk inflation criterion for multiple regression.
\newblock \emph{Ann. Statist.}, 22:\penalty0 1947--1975, 1994.

\bibitem[Friedman et~al.(2010)Friedman, Hastie, and Tibshirani]{friedman2010}
J.~Friedman, T.~Hastie, and R.~Tibshirani.
\newblock Regularization paths for generalized linear models via coordinate
  descent.
\newblock \emph{Journal of statistical software}, 33\penalty0 (1):\penalty0
  1--22, 2010.

\bibitem[Huang and Ma(2010)]{Huang2010Variable}
J.~Huang and S.~Ma.
\newblock Variable selection in the accelerated failure time model via the
  bridge method.
\newblock \emph{Lifetime Data Analysis}, 16\penalty0 (2):\penalty0 176--95,
  2010.

\bibitem[Huang et~al.(2006)Huang, Ma, and Xie]{Huang2005}
J.~Huang, S.~Ma, and H.~Xie.
\newblock Regularized estimation in the accelerated failure time model with
  high-dimensional covariates.
\newblock \emph{Biometrics}, 62\penalty0 (3):\penalty0 813--820, 2006.

\bibitem[Johnson(2009)]{Johnson2009On}
B.~A. Johnson.
\newblock On lasso for censored data.
\newblock \emph{Electronic Journal of Statistics}, 3\penalty0 (2009):\penalty0
  485--506, 2009.

\bibitem[Johnson et~al.(2008)Johnson, Lin, and Zeng]{Johnson2008Penalized}
B.~A. Johnson, D.~Y. Lin, and D.~Zeng.
\newblock Penalized estimating functions and variable selection in
  semiparametric regression models.
\newblock \emph{Journal of the American Statistical Association}, 103\penalty0
  (482):\penalty0 672--680, 2008.

\bibitem[Kalbfleisch and Prentice(2002)]{Kalbfleisch2002}
J.~D. Kalbfleisch and R.~L. Prentice.
\newblock \emph{The Statistical Analysis of Failure Time Data, 2nd Edition}.
\newblock 2002.

\bibitem[Kawaguchi et~al.(2019)Kawaguchi, Suchard, Liu, and Li]{Eric}
E.~S. Kawaguchi, M.~A. Suchard, Z.~Liu, and G.~Li.
\newblock A surrogate l0 sparse cox's regression with applications to sparse
  high-dimensional massive sample size time-to-event data.
\newblock \emph{Statistics in Medicine}, n/a\penalty0 (n/a), 2019.
\newblock \doi{10.1002/sim.8438}.
\newblock URL \url{https://onlinelibrary.wiley.com/doi/abs/10.1002/sim.8438}.

\bibitem[Koul et~al.(1981)Koul, Susarla, and Ryzin]{Koul1981Regression}
H.~Koul, V.~Susarla, and J.~V. Ryzin.
\newblock Regression analysis with randomly right-censored data.
\newblock \emph{Annals of Statistics}, 9\penalty0 (6):\penalty0 1276--1288,
  1981.

\bibitem[Leurgans(1987)]{Leurgans1987}
S.~Leurgans.
\newblock Linear models, random censoring and synthetic data.
\newblock \emph{Biometrika}, 74\penalty0 (2):\penalty0 301--309, 1987.

\bibitem[Lin et~al.(2010)Lin, Foster, and Ungar]{lin2010risk}
D.~Lin, D.~P. Foster, and L.~H. Ungar.
\newblock A risk ratio comparison of l0 and l1 penalized regressions.
\newblock \emph{University of Pennsylvania, techical report}, 2010.

\bibitem[Liu et~al.(2019)Liu, Chen, and Li]{YiLiu}
Y.~Liu, X.~Chen, and G.~Li.
\newblock A new joint screening method for right-censored time-to-event data
  with ultra-high dimensional covariates.
\newblock \emph{Statistical Methods in Medical Research}, 2019.
\newblock \doi{10.1177/0962280219864710}.

\bibitem[Mallows(1973)]{mallows1973some}
C.~Mallows.
\newblock Some comments on $c_p$.
\newblock \emph{Technometrics}, 15:\penalty0 661--675, 1973.

\bibitem[Nachmani et~al.(2019)Nachmani, Bothmer, Grisendi, Mele, and
  Pandolfi]{Daphna2019Germline}
D.~Nachmani, A.~H. Bothmer, S.~Grisendi, A.~Mele, and P.~P. Pandolfi.
\newblock \emph{Nature Genetics}, 51\penalty0 (10):\penalty0 1518--1529, 2019.

\bibitem[Schwarz(1978)]{schwarz1978estimating}
G.~Schwarz.
\newblock Estimating the dimension of a model.
\newblock \emph{Ann. Statist.}, 6:\penalty0 461--464, 1978.

\bibitem[Shen et~al.(2012)Shen, Pan, and Zhu]{shen2012likelihood}
X.~Shen, W.~Pan, and Y.~Zhu.
\newblock Likelihood-based selection and sharp parameter estimation.
\newblock \emph{J. Amer. Statist. Assoc.}, 107:\penalty0 223--232, 2012.

\bibitem[Stute(1993)]{stute1993consistent}
W.~Stute.
\newblock Consistent estimation under random censorship when covariables are
  present.
\newblock \emph{Journal of Multivariate Analysis}, 45\penalty0 (1):\penalty0
  89--103., 1993.

\bibitem[Tibshirani(1996)]{tib1996}
R.~Tibshirani.
\newblock Regression shrinkage and selection via the lasso.
\newblock \emph{Journal of the Royal Statistical Society. Series B
  (Methodological)}, 58\penalty0 (1):\penalty0 267--288, 1996.

\bibitem[Tibshirani(1997)]{Tibshirani1997}
R.~Tibshirani.
\newblock The lasso method for variable selection in the cox model.
\newblock \emph{Statistics in Medicine}, 16\penalty0 (4):\penalty0 385--395,
  1997.

\bibitem[Wang et~al.(2008)Wang, Nan, Zhu, and Beer]{Wang2008Doubly}
S.~Wang, B.~Nan, J.~Zhu, and D.~G. Beer.
\newblock \emph{Biometrics}, 64\penalty0 (1):\penalty0 132--40, 2008.

\bibitem[Yi and Li(2009)]{Yi2009Dantzig}
Yi and Li.
\newblock Dantzig selector for censored linear regression models:with
  applications in high dimensional data analysis.
\newblock In \emph{International Conference on Financial Statistics and
  Financial Econometrics}, 2009.

\bibitem[Zhang(2010)]{zhang2010}
C.-H. Zhang.
\newblock Nearly unbiased variable selection under minimax concave penalty.
\newblock \emph{The Annals of statistics}, 38\penalty0 (2):\penalty0 894--942,
  2010.

\bibitem[Zhao et~al.(2019)Zhao, Wu, Li, and Sun]{Jianguo2019}
H.~Zhao, Q.~Wu, G.~Li, and J.~Sun.
\newblock Simultaneous estimation and variable selection for interval-censored
  data with broken adaptive ridge regression.
\newblock \emph{Journal of the American Statistical Association}, 0\penalty0
  (0):\penalty0 1--13, 2019.
\newblock \doi{10.1080/01621459.2018.1537922}.
\newblock URL \url{https://doi.org/10.1080/01621459.2018.1537922}.

\bibitem[Zhou(1992)]{Zhou1992Asymptotic}
M.~Zhou.
\newblock Asymptotic normality of the synthetic data regression estimator for
  censored survival data.
\newblock \emph{Annals of Statistics}, 20\penalty0 (2):\penalty0 1002--1021,
  1992.

\bibitem[Zhu et~al.(2011)Zhu, Li, Li, and Zhu]{Zhu2011Model}
L.~Zhu, L.~Li, R.~Li, and L.~Zhu.
\newblock Model-free feature screening for ultrahigh dimensional data.
\newblock \emph{Publications of the American Statistical Association},
  106\penalty0 (496):\penalty0 1464--1475, 2011.

\bibitem[Zou(2006)]{zou2006}
H.~Zou.
\newblock The adaptive lasso and its oracle properties.
\newblock \emph{Journal of the American statistical association}, 101\penalty0
  (476):\penalty0 1418--1429, 2006.

\end{thebibliography}

\begin{thebibliography}{37}
\providecommand{\natexlab}[1]{#1}
\providecommand{\url}[1]{{#1}}
\providecommand{\urlprefix}{URL }
\expandafter\ifx\csname urlstyle\endcsname\relax
  \providecommand{\doi}[1]{DOI~\discretionary{}{}{}#1}\else
  \providecommand{\doi}{DOI~\discretionary{}{}{}\begingroup
  \urlstyle{rm}\Url}\fi
\providecommand{\eprint}[2][]{\url{#2}}

\bibitem[{Akaike(1974)}]{akaike1974new}
Akaike, H. (1974). A new look at the statistical model identification. IEEE Transactions on Automatic Control, 19, 716--723.

\bibitem[{Breheny and Huang(2011)}]{breheny2011}
Breheny, P. and Huang, J. (2011) Coordinate descent algorithms for nonconvex penalized regression with applications to biological feature selection. The annals of
  applied statistics, 5(1), 232--253.

\bibitem[{Breiman(1996)}]{breiman1996heuristics}
Breiman, L. (1996). Heuristics of instability and stabilization in model
  selection. Ann Statist, 24, 2350--2383.

\bibitem[{Buckley and James(1979)}]{Buckley1979Linear}
Buckley, J. and James, I. (1979). Linear regression with censored data. Biometrika, 66(3), 429--436.

\bibitem[{Cai et~al.(2009)Cai, Huang, and Tian}]{Cai2009Regularized}
Cai, T., Huang, J. and Tian, L. (2009). Regularized estimation for the accelerated failure time model. Biometrics, 65(2), 394--404.

\bibitem[{Chen and Chen(2008)}]{chen2008extended}
Chen, J. and Chen, Z. (2008). Extended bayesian information criteria for model selection with large model spaces. Biometrika, 95, 759--771

\bibitem[{Cox(1972)}]{Cox1972}
Cox, D.R. (1972). Regression models and life-tables. Journal of the Royal Statistical Society: Series B (Methodological), 34(2):187--220.

\bibitem[{Cui et~al.(2015)Cui, Li, and Zhong}]{Cui2015Model}
Cui, H., Li, R. and Zhong, W. (2015). Model-free feature screening for ultrahigh dimensional discriminant analysis. Journal of the American Statistical Association, 110(510), 630--641.

\bibitem[{Dai et~al.(2018{\natexlab{a}})Dai, Chen, and Li}]{dai2018fusedBAR}
Dai, L., Chen, K. and Li, G. (2018{\natexlab{a}}). The broken adaptive ridge procedure and its applications. Statistica Sinica, \doi{10.5705/ss.202018.0075}.

\bibitem[{Dai et~al.(2018{\natexlab{b}})Dai, Chen, Sun, Liu, and
  Li}]{DAI2018334}
Dai, L., Chen, K., Sun, Z., Liu, Z. and Li, G. (2018{\natexlab{b}}). Broken adaptive ridge regression and its asymptotic properties. Journal of Multivariate Analysis, 168, 334--351.

\bibitem[{Dickson et~al.(1989)Dickson, Grambsch, Fleming, Fisher, and
  Langworthy}]{DicksonPrognosis}
Dickson, E.R., Grambsch, P.M., Fleming, T.R., Fisher, L.D. and Langworthy, A. (1989). Prognosis in primary biliary cirrhosis: Model for decision making. Hepatology, 10(1), 1--7.

\bibitem[{Fan and Li(2001)}]{fan2001}
Fan, J. and Li, R. (2001). Variable selection via nonconcave penalized likelihood and its oracle properties. Journal of the American statistical Association, 96(456), 1348--1360.

\bibitem[{Fan and Li(2002)}]{Fan2002}
Fan, J. and Li, R. (2002). Variable selection for cox's proportional hazards model and frailty model. The Annals of Statistics, 30(1), 74--99.

\bibitem[{Fan and Lv(2008)}]{Fan2008Sure}
Fan, J. and Lv, J. (2008). Sure independence screening for ultrahigh dimensional feature space. Journal of the Royal Statistical Society, 70(5), 849--911.

\bibitem[{Fleming and Harrington(1991)}]{Fleming1991Counting}
Fleming, T.R. and Harrington, D.P. (1991). Counting Processes and Survival Analysis. John Wiley \& Sons.

\bibitem[{Foster and George(1994)}]{foster1994risk}
Foster, D. and George, E. (1994). The risk inflation criterion for multiple regression. Ann Statist, 22, 1947--1975.

\bibitem[{Friedman et~al.(2010)Friedman, Hastie, and Tibshirani}]{friedman2010}
Friedman, J., Hastie, T. and Tibshirani, R. (2010). Regularization paths for generalized linear models via coordinate descent. Journal of statistical software, 33(1), 1--22.

\bibitem[{Jin et~al.(2006)Jin, Lin, and Ying}]{Jin2006On}
Jin, Z., Lin, D.Y. and Ying, Z. (2006). On least-squares regression with censored data. Biometrika, 93(1), 147--161.

\bibitem[Hu, J. , & Chai, H. . (2013)] {HuChai2013}
Adjusted regularized estimation in the accelerated failure time model with high dimensional covariates. Journal of Multivariate Analysis, 122(Complete), 96-114.

\bibitem[Johnson(2009)]{Johnson2009On}
B.~A. Johnson.
\newblock On lasso for censored data.
\newblock \emph{Electronic Journal of Statistics}, 3\penalty0 (2009):\penalty0
  485--506, 2009.
\bibitem[{Kalbfleisch and Prentice(2002)}]{Kalbfleisch2002}
Kalbfleisch, J.D. and Prentice, R.L. (2002). The Statistical Analysis of Failure Time Data, 2nd Edition.

\bibitem[{Kawaguchi et~al.(2019)Kawaguchi, Suchard, Liu, and Li}]{Eric}
Kawaguchi, E.S., Suchard, M.A., Liu, Z. and Li, G. (2019). A surrogate l0 sparse cox's regression with applications to sparse high-dimensional massive sample size time-to-event data. Statistics in Medicine n/a(n/a), \doi{10.1002/sim.8438},
\eprint{https://onlinelibrary.wiley.com/doi/pdf/10.1002/sim.8438}.

\bibitem[{Koul et~al.(1981)Koul, Susarla, and Ryzin}]{Koul1981Regression}
Koul, H., Susarla, V. and Ryzin, J.V. (1981). Regression analysis with randomly right-censored data. Annals of Statistics, 9(6), 1276--1288.

\bibitem[{Leurgans(1987)}]{Leurgans1987}
Leurgans, S. (1987). Linear models, random censoring and synthetic data.
  Biometrika, 74(2), 301--309.

\bibitem[{Liu et~al.(2019)Liu, Chen, and Li}]{YiLiu}
Liu, Y., Chen, X. and Li, G. (2019). A new joint screening method for right-censored time-to-event data with ultra-high dimensional covariates. Statistical Methods in Medical Research, \doi{10.1177/0962280219864710}.

\bibitem[{Mallows(1973)}]{mallows1973some}
Mallows, C. (1973). Some comments on $c_p$. Technometrics, 15, 661--675.

\bibitem[{Miller and Halpern(1982)}]{MILLER1982Regression}
Miller, R. and Halpern, J. (1982). Regression with censored data. Biometrika, 69(3), 521--531.

\bibitem[{Prentice(1978)}]{Prentice1978Linear}
Prentice, R.L. (1978). Linear rank tests with right censored data. Biometrika, 65(1), 167--179.

\bibitem[{Schwarz(1978)}]{schwarz1978estimating}
Schwarz, G. (1978). Estimating the dimension of a model. Ann Statist, 6, 461--464.

\bibitem[{Shen et al. (2012)}]{shen2012likelihood}
Shen, X., Pan, W. and Zhu, Y. (2012). Likelihood-based selection and sharp parameter estimation. Journal of the American statistical Association, 107, 223--232.

\bibitem[{Slutsky (1926)}]{Slutsky1926}
Slutsky, E. (1926). Uber stochastische asymptoten und grenzwerte. Tohoku
  Mathematical Journal, First Series, 27, 67--70.

\bibitem[{Stute(1993)}]{stute1993consistent}
Stute, W. (1993). Consistent estimation under random censorship when covariables are present. Journal of Multivariate Analysis, 45(1), 89--103.

\bibitem[{Tibshirani(1996)}]{tib1996}
Tibshirani, R. (1996). Regression shrinkage and selection via the lasso. Journal of the Royal Statistical Society: Series B (Methodological), 58(1), 267--288.

\bibitem[{Tibshirani(1997)}]{Tibshirani1997}
Tibshirani, R. (1997). The lasso method for variable selection in the cox model. Statistics in Medicine, 16(4), 385--395.

\bibitem[{Zhang(2010)}]{zhang2010}
Zhang, C.H. (2010). Nearly unbiased variable selection under minimax concave penalty. The Annals of statistics, 38(2), 894--942.

\bibitem[{Zhao et~al.(2019)Zhao, Wu, Li, and Sun}]{Jianguo2019}
Zhao, H., Wu, Q., Li, G. and Sun, J. (2019). Simultaneous estimation and variable selection for interval-censored data with broken adaptive ridge regression. Journal of the American Statistical Association 0(0):1--13,
  \doi{10.1080/01621459.2018.1537922}.

\bibitem[{Zhou(1992)}]{Zhou1992Asymptotic}
Zhou, M. (1992). Asymptotic normality of the synthetic data regression estimator for censored survival data. Annals of Statistics, 20(2), 1002--1021.

\bibitem[{Zhu et~al.(2011)Zhu, Li, Li, and Zhu}]{Zhu2011Model}
Zhu, L.P., Li, L., Li, R. and Zhu, L.X. (2011) Model-free feature screening for ultrahigh dimensional data. Journal of the American Statistical Association, 106(496), 1464--1475.

\bibitem[{Zou(2006)}]{zou2006}
Zou, H. (2006). The adaptive lasso and its oracle properties. Journal of the American statistical association, 101(476), 1418--1429.

\end{thebibliography}

\appendix
\section{Proofs of the theorem}
\label{sec: appdx}
We first introduce notations and lemmas used to prove Theorem \ref{theorem:CBAR}.
\par Using \citet{Leurgans1987} method, we transform $\*Y$ into synthetic data $\*Y^*$.
Let $ \bbeta=( \bm\alpha^\top, \bm\gamma^\top)^\top$, where $ \bm\alpha$ and $ \bm\gamma$ are $q_n\times 1$ and $(p_n-q_n) \times 1$ vector respectively, $\*\Sigma_n=\*X^\top \*X/n$.
\begin{equation}
\label{eq:t}
     g( \bbeta)=\{ \*X^\top \*X+\lambda_n \*D( \bbeta)\}^{-1}\* X^\top \*Y^* =( \bm\alpha^{*}( \bbeta)^\top, \bm\gamma^*( \bbeta)^\top)^\top.
\end{equation}
For simplicity, we write $ \bm\alpha^*( \bbeta)$ and $ \bm\gamma^*( \bbeta)$ as $ \bm\alpha^*$ and $ \bm\gamma^*$ hereafter. $\*\Sigma_n^{-1}$ can be partitioned as
\[ \*\Sigma_n^{-1}=\begin{pmatrix}
 \*A_{11} &  \*A_{12}\\
 \*A^\top_{12} &  \*A_{22}
\end{pmatrix}
\]
where the $A_{11} $ is a {$q\times q$ }matrix.
Multiplying $( \*X^\top \*X)^{-1}( \*X^\top \*X+\lambda_n \*D( \bbeta))$ to equation (\ref{eq:t})
\begin{equation}
\label{eq:t1}
\begin{pmatrix}
   {\bm\alpha}^*-{\bm\beta}_{01}\\
   {\bm\gamma}^*
  \end{pmatrix}
+\frac{\lambda_n}{n}\begin{pmatrix}
\*A_{11}\*D_1(\bm\alpha)\bm\alpha^*+\*A_{12}\*D_2(\bm\gamma)\bm\gamma^*\\
\*A_{12}^\top\*D_1(\bm\alpha)\bm\alpha^*+\*A_{22}\*D_2(\bm\gamma)\bm\gamma^*
\end{pmatrix}=(\*X^\top\*X)^{-1}\*X^\top \bm\varepsilon^* {=\hat{\bbeta}_{\rm Z}-\bbeta_0},
\end{equation}
where $\bm\varepsilon^*=\*Y^*-\*X{\bbeta_0}$, {$\hat\bbeta_{\rm Z}=(\*X^\top\*X)^{-1}\*X^\top\*Y^*$}, $ \*D_1( \bm\alpha)=\mbox{diag }(\alpha_1^{-2},...,\alpha_{{q}}^{-2})$ and $ \*D_2( \bm\gamma)=\mbox{diag }(\gamma_1^{-2},...,\gamma_{p_n-{q}}^{-2})$.
\begin{lemma}
	\label{lemmma:1}
	{ Let $\delta_n$ be a sequence of positive real numbers satisfying $\delta_n \to  \infty$ and $p_n\delta_n^2/\lambda_n \to  0$. }
	Define \textcolor{black}{$\*H_n = \{ \bbeta\in \mathbb{R}^{p_n}: \|\bm\beta-\bm\beta_0\| \leq \delta_n\sqrt{p_n/n}\}$ and $\*H_{n1} = \{ \bm\alpha\in \mathbb{R}^{{q}}: \|\bm\alpha-\bm\beta_{01}\| \leq \delta_n\sqrt{p_n/n}\}$}.  Assume conditions (C1)-(C5) hold. Then, with probability tending to $1$, we have
\begin{itemize}
	\item [(a)]
	$\sup_{ \bm\beta \in  \*H_n} {\| \bm\gamma^*\|}/{\| \bm\gamma\|}< {1}/{C_0},\mbox{ for some constant } C_0>1$;
	\item [(b)]
	$g $ is a mapping from $ \*H_n$ to itself.
\end{itemize}
\end{lemma}

\begin{proof} We first prove part (a).

{First, under
$\lambda_n/\sqrt{n} \to  0$ and $p_n\delta_n^2/\lambda_n \to  0$, we have $\delta_n\sqrt{p_n/n} \to  0$.}

Let $\hat\bbeta_{\rm Z}=(\*X^\top\*X)^{-1}\*X^\top\*Y^*$,
$\omega_{ji}=((\*X^\top \*X)^{-1}\*X^\top )_{ji}$, $\mu_j^*=\sum_i \omega_{ji}\int_{0}^{T_n}{F_i dt}$ and $\bm \mu =(\mu_1^*, \mu_2^*, ..., \mu_{{pn}}^*)$.
{For any $p_n$-vector $\*b_n$ which $\| \*b_n\|\leq 1$, define $t_n^2= \*b_n^\top \*\Omega(\infty)) \*b_n$.
Then, we have $\sqrt{n} \, t_n^{-1} \*b_n^\top (\hat\bbeta_{\rm Z}-\bm\mu) \rightarrow_D N(0,1).$
This result can be proved using similar techniques to those used in the proof of Theorem 3.1 of \citet{Zhou1992Asymptotic} along the same lines  as outlined below: First, we separate $\*b_n^\top (\hat\bbeta_{\rm Z}-\bm\mu)$ like (3.6) in \citet{Zhou1992Asymptotic} with a main term $S_{\bbeta}(T^n)$ and a remainder term $SS_{\bbeta}(T^n)$, i.e., $\*b_n^\top (\hat\bbeta_{\rm Z}-\bm\mu)=S_{\bbeta}(T^n)+SS_{\bbeta}(T^n)$, where $S_{\bbeta}(T^n)$ is a weighted sum of $\hat H(t)-H(t)$ and  $\hat G(t)-G(t)$; and $SS_{\bbeta}(T^n)$ is a weighted sum of $(\hat H(t)-H(t))(\hat G(t)-G(t))$ and $(\hat H(t)-H(t))(\hat H(t)-H(t))$. Second, under conditions (C2) and (C3), one can show that $\sqrt{n}SS_{\bbeta}(T^n)$ is negligible. Finally, by applying the martingale central limit theorem and conditions (C1) and (C4), we  establish the asymptotic normality of $\sqrt{n}S_{\bbeta}(T^n)$. 
{ By condition (C1) and (C2), we have
$ \sqrt{n}t_n^{-1}\*b_n^\top(\bbeta_0 - \bm\mu) = o_p(1)$, for $\*b_n=\*e_i=(0,...,1,0,...,0)$}. Hence,
we have
 $\|\hat{\bbeta}_{\rm Z}-\bbeta_0\|^2=O_p(p_n/n)$.}

It then follows from (\ref{eq:t1}) that
	\begin{equation}
	\label{super:3}
	\sup_{ \beta \in  \*H_n}\big\|  \bm\gamma^*+ \lambda_n  \*A_{12}^\top  \*D_1( \bm\alpha) \bm\alpha^*/n + \lambda_n  \*A_{22} \*D_2( \bm\gamma) \bm\gamma^*/n \big\|=O_p(\sqrt{{p_n}/{n}}).
	\end{equation}
	Note that \textcolor{black}{$\|\bm\alpha - \bm\beta_{01}\|\leq \delta_n(p_n/n)^{1/2}$ and $\| \bm\alpha^*\|\leq \|g( \bbeta)\|\leq \|\hat{ \bbeta}_{\rm Z}\|=O_p({\sqrt{p_n}})$}. By assumptions (C4) and (C5), we have
	\begin{equation}
   \label{super:4}
    \begin{split}
	\sup_{ \bbeta \in  \*H_n}\left\| \lambda_n \*A_{12}^\top  \*D_1( \bm\alpha) \bm\alpha^*/n \right\|  & \leq \frac{\lambda_n}{n} \, \|  \*A_{12}^\top \|\sup_{ \bbeta \in  \*H_n}\| \*D_1( \bm\alpha)
	\bm\alpha^*\|\\ & \leq{\color{black} \sqrt{2}\, \tilde{C}\, \frac{\lambda_n}{n} \, {\frac{a_{1}}{a_{0}^2}}\sup_{ \bbeta \in  \*H_n}\| \bm\alpha^*\|}= o_p(\sqrt{{p_n}/{n}}),
    \end{split}
	\end{equation}
	where the second inequality uses the fact $\|\*A_{12}^\top \|\leq \sqrt{2}\, \tilde{C}$, which follows from the inequality $\|\*A_{12}\*A_{12}^\top \|-\|\*A_{11}^2\|\leq \|\*A_{11}^2+\*A_{12}\*A_{21}\|\leq\|\*\Sigma_n^{-2}\|<\tilde{C}^2.$
	Combining (\ref{super:3}) and  (\ref{super:4}) gives
	\begin{equation}
	\label{superr:1}
	\sup_{ \bbeta \in  \*H_n}\left\|  \bm\gamma^*+ \lambda_n \*A_{22} \*D_2( \bm\gamma) \bm\gamma^*/n \right\|=O_p(\sqrt{{p_n}/{n}}).
	\end{equation}
	Note that $ \*A_{22}=\sum_{i=1}^{p_n-{q}}\tau_{2i} \*u_{2i} \*u_{2i}^\top $ is positive definite and by the singular value decomposition, , where $\tau_{2i}$ and $ \*u_{2i}$ are eigenvalues and eigenvectors of $ \*A_{22}$. Then, since $1/\tilde{C} <\tau_{2i}< \tilde{C}$, we have
	\begin{equation}\nonumber
	\begin{split}
	\frac{\lambda_n}{n} \, \| \*A_{22} \*D_2( \bm\gamma) \bm\gamma^*\|&=\frac{\lambda_n}{n}\left\|\sum_{i=1}^{p_n-{q}}\tau_{2i} \*u_{2i} \*u_{2i}^\top  \*D_2( \bm\gamma) \bm\gamma^*\right\|
	= \frac{\lambda_n}{n}\left\{\sum_{i=1}^{p_n-{q}}\tau_{2i}^2\| \*u_{2i}^\top  \*D_2( \bm\gamma) \bm\gamma^*\|^2\right\}^{1/2}\\
	&\geq \frac{\lambda_n}{n}\frac{1}{\tilde{C}}\left\{\sum_{i=1}^{p_n-{q}}\| \*u_{2i}^\top  \*D_2( \bm\gamma) \bm\gamma^*\|^2\right\}^{1/2} = \frac{1}{\tilde{C}} \left\| \lambda_n  \*D_2( \bm\gamma) \bm\gamma^* /n \right\|.
	\end{split}
	\end{equation}
	This, together with (\ref{superr:1}) and (C4), implies that {with probability tending to $1$,}
	\begin{equation}
	\label{superr:2}
	\frac{1}{\tilde{C}}\left\|\lambda_n  \*D_2( \bm\gamma) \bm\gamma^*/n \right\|-\| \bm\gamma^*\|\leq \delta_n\sqrt{{p_n}/{n}}.
	\end{equation}
	\par Let $ \*d_{\gamma*/\gamma}=(\gamma^*_1/\gamma_1, \ldots,\gamma^*_{p_n-{q}}/\gamma_{p_n-{q}})^\top $. Because $\| \bm\gamma\|\leq\delta_n\sqrt{p_n/n}$,  we have
	\begin{equation}
	\label{superr:21}
	\frac{1}{\tilde{C}}\left\|\frac{\lambda_n}{n}\, \*D_2( \bm\gamma) \bm\gamma^*\right\| =\frac{1}{\tilde{C}}\frac{\lambda_n}{n}\left\|{ \{\*D_2( \bm\gamma)}\}^{1/2} \*d_{\bm\gamma*/\bm\gamma}\right\| \geq \frac{1}{\tilde{C}}\frac{\lambda_n}{n}\frac{\sqrt{n}}{\delta_n\sqrt{p_n}} \, \| \*d_{\bm\gamma*/\bm\gamma}\|
	\end{equation}
	and
	\begin{equation}
	\label{superr:22}
	\| \bm\gamma^*\| =\|{ \*D_2( \bm\gamma)}^{-1/2} \*d_{\bm\gamma*/\bm\gamma}\|\leq \frac{\delta_n\sqrt{p_n}}{\sqrt{n}}\, \| \*d_{\bm\gamma*/\bm\gamma}\|.
	\end{equation}
	Combining (\ref{superr:2}), (\ref{superr:21}) and (\ref{superr:22}), we have
	 that with probability tending to $1$,
	\begin{equation}
	\label{superr:23}
	\| \*d_{\bm\gamma*/\bm\gamma}\|\leq \frac{1}{{\lambda_n}/({p_n}\delta_n^2 \tilde{C})-1}<{1}/{C_0}
	\end{equation}
	{for some constant $C_0 > 1$ provided that $\lambda_n/({p_n}\delta_n^2) \to \infty$}.

It is worth noting that $\Pr(\| \*d_{\bm\gamma*/\bm\gamma}\| \to  0) \to  1$, as $n \to  \infty$.
	Furthermore,  with probability tending to $1$,
	\begin{equation*}
	\label{superr:24}
	\| \bm\gamma^*\|\leq\| \*d_{\bm\gamma^*/\bm\gamma}\|\max_{1\leq j\leq (p_n-{q})}|\bm\gamma_j|\leq\| \*d_{\bm\gamma^*/\bm\gamma}\|\times\|\bm\gamma\|\leq \| \bm\gamma\|/C_0.
	\end{equation*}
	This proves part (a).

	Next we prove part (b). First, it is easy to see from (\ref{superr:22}) and (\ref{superr:23}) that, as $n \to \infty$,
	\begin{equation}
	\label{lemma1:i1}
	\Pr  \Big (\|
	\bm\gamma^*\|\leq \delta_n\sqrt{p_n/n} \Big ) \to  1.
	\end{equation}
	Then, by (\ref{eq:t1}), we have
	\begin{equation}
	\label{super:6}
	\sup_{ \bm\beta \in  \*H_n}\left\|  \bm\alpha^*- \bbeta_{01}+ \lambda_n  \*A_{11} \*D_1( \bm\alpha) \bm\alpha^*/n +\lambda_n  \*A_{12} \*D_2( \bm\gamma)\bm\gamma^* /n \right\|=O_p(\sqrt{{p_n}/{n}}).
	\end{equation}
	Similar to (\ref{super:4}), it is easily to verify that
	\begin{equation}
	\label{super:61}
	\sup_{ \bbeta \in  \*H_n}\left\|  \lambda_n  \*A_{11} \*D_1( \bm\alpha)\bm\alpha^*/n \right\| = o_p(\sqrt{{p_n}/{n}}).
	\end{equation}
	Moreover, with probability tending to $1$,
	\begin{equation}
	\label{super:62}
	\sup_{ \bbeta \in  \*H_n} \left\| \lambda_n  \*A_{12} \*D_2( \bm\gamma) \bm\gamma^*/n \right\| \leq\frac{\lambda_n}{n} \sup_{ \bbeta \in \* H_n} \left\| \*D_2( \bm\gamma) \bm\gamma^*\right\| \times \| \*A_{12}\|\leq 2\sqrt{2}\tilde{C}^2\delta_n\sqrt{{p_n}/{n}},
	\end{equation}
	where the last step follows from (\ref{superr:2}),  (\ref{lemma1:i1}), and the fact that $\|\* A_{12}\|\leq \sqrt{2}\tilde{C}$. It follows from (\ref{super:6}), (\ref{super:61}) and (\ref{super:62}) that
	with probability tending to $1$,
	\begin{equation}
	\label{superr:3}
	\sup_{ \bbeta \in  \*H_n} \|  \bm\alpha^*- \bbeta_{01}\|\leq { \big (2\sqrt{2}\tilde{C}^2+1 \big )\delta_n n^{-1/2}\sqrt{p_n}}.
	\end{equation}
	Because $\delta_n\sqrt{p_n}/\sqrt{n} \to  0$, we have, as $n \to \infty$,
	{\color{black}\begin{equation}
	\label{lemma1:i2}
	\Pr ( \bm\alpha^*\in \*H_{n1}) \to  1.
	\end{equation}}
	Combining (\ref{lemma1:i1}) and (\ref{lemma1:i2}) completes the proof of part (b).
\end{proof}
\begin{lemma}
	\label{lemmma:2}
	Assume that (C1)-(C5) hold.
    For any ${q}$-vector $ {\*c}$ satisfying $\| {\*c}\|\leq 1$, define ${z^2}= {\*c}^\top \*\Omega_{1} {\*c}$ as in Theorem \ref{theorem:CBAR}.
	Define
	\begin{equation}\label{eq:7}
	f( \bm\alpha)=\{ \*X_{1} ^\top   \*X_1+\lambda_n \*D_1( \bm\alpha)\}^{-1} \*X_1^\top  \*Y^*.
	\end{equation}
	Then, with probability tending to $1$,	

	(a) $f( \bm\alpha)$ is a contraction mapping from {\color{black}$\*B_{n} \equiv\{ \bm\alpha\in \mathbb{R}^{{q}}: \|\bm\alpha-\bm\beta_{01}\| \leq \delta_n\sqrt{p_n/n}\}$} to itself;	

	(b)
	$\sqrt{n} \, {z^{-1} \*c^\top}(\hat{ \bm\alpha}^{\circ}- \bbeta_{01})  \rightsquigarrow \mathcal{N}(0,1),$
	where $\hat{ \bm\alpha}^{\circ}$ is the unique fixed point of $f(\bm\alpha)$ defined by
$$
	\hat{ \bm\alpha}^{\circ}= \{ \*X_{1}^\top { \*X_1}+\lambda_n \*D_1( \hat{ \bm\alpha}^{\circ})\}^{-1}{ \*X}_1^\top \*Y^*.
$$	
\end{lemma}

\begin{proof}
We first prove part (a). Note that  (\ref{eq:7}) can be rewritten  as
\begin{equation}\nonumber
f( \bm\alpha)- \bbeta_{01}+\frac{\lambda_n}{n} \bm\Sigma_{n1}^{-1} \*D_1( \bm\alpha)f( \bm\alpha) = {\hat{\bbeta}_{1\rm Z}- \bbeta_{01}}.
\end{equation}
{where $\hat{\bbeta}_{1\rm Z}=( \*X_1^\top \*X_1)^{-1} \*X_1^\top \* Y^*$}.
Then,
\begin{equation}
\label{lemma2:m1}
\sup_{ \bm\alpha\in \*B_n}\left\|f( \bm\alpha)- \bbeta_{01}+(\lambda_n/n) \*\Sigma_{n1}^{-1} \*D_1( \bm\alpha)f( \bm\alpha) \right\|= O_p({1/\sqrt{n}}).
\end{equation}
\begin{equation}
\label{lemma2:m2}
\sup_{ \bm\alpha\in \*B_n}\left\|(\lambda_n/n) \*\Sigma_{n1}^{-1} \*D_1( \bm\alpha)f( \bm\alpha) \right\| =o_p({1/\sqrt{n}}).
\end{equation}
It follows from (\ref{lemma2:m1}) and (\ref{lemma2:m2}) that
\begin{equation}
\label{lemma2:m3}
\sup_{ \bm\alpha\in \*B_n}\left\|f( \bm\alpha)- \bbeta_{01} \right\|\leq \delta_n{/\sqrt{n}},
\end{equation}
where $\delta_n \to\infty$ and $\delta_n{/\sqrt{n}}\rightarrow 0$. Then we can get
\begin{equation}
\label{map}
\mbox{Pr}(f( \bm\alpha)\in \*B_n)\rightarrow 1, \mbox{  as } n\rightarrow \infty.
\end{equation}
This means that $f$ is a mapping from the region $\*B_n$ to itself.

Rewrite (\ref{eq:7}) as
$\{ \*X_{1}^\top{ \*X_1}+\lambda_n \*D_1( \bm\alpha)\}f( \bm\alpha)= \*X_{1}^\top \*Y^*$,
then, we have
\begin{equation}
( \*\Sigma_{n1}+(\lambda_n/n){ \*D_1}(\bm \alpha))\dot{f}( \bm\alpha)+(\lambda_n/n)\mbox{ diag } \{-2f_j( \bm\alpha)/{\alpha_j^3}\} ={ 0},
\end{equation}
where $\dot{f}( \bm\alpha)={\partial f( \bm\alpha)}/{\partial { \bm\alpha^\top}}$ and $\mbox{ diag } \{\frac{-2f_j( \bm\alpha)}{\alpha_j^3}\}= \mbox{ diag }\{\frac{-2f_1( \bm\alpha)}{\alpha_1^3},...,\frac{-2f_{{q}}( \bm\alpha)}{\alpha_{{q}}^3}\}.$
With the assumption $\lambda_n/\sqrt{n}\rightarrow 0$,
\begin{equation}
\label{eq:a}
\sup_{ \bm\alpha\in \*B_n }\|\{ \*\Sigma_{n1}+\frac{\lambda_n}{n}{ \*D}_1( \bm\alpha)\}\dot{f}( \bm\alpha)\| =  \sup_{ \bm\alpha\in \*B_n} \frac{2\lambda_n}{n}\|\mbox{ diag } \{\frac{f_j( \bm\alpha)}{\alpha_j^3}\}\|=o_p(1).
\end{equation}
Write $ \*\Sigma_{n1} = \sum_{i=1}^{{q}}\tau_{1i} \*u_{1i} \*u_{1i}^\top$, where $\tau_{1i}$ and $ \*u_{1i}$ are eigenvalues and eigenvectors of $ \*\Sigma_{n1}$. Then,  by (C4), $1/\tilde{C}<\tau_{1i}< \tilde{C}$ for all $i$  and
\begin{equation}
\label{superr:38}
\begin{split}
\| \*\Sigma_{n1}\dot{f}( \bm\alpha)\|&=\sup_{\| \*x\|=1, \*x\in R^{{q}}}\| \*\Sigma_{n1}\dot{f}( \bm\alpha) \*x\|=\sup_{\| \*x\|=1, \*x\in R^{{q}}}\left\|\sum_{i=1}^{{q}}\lambda_{1i} \*u_{1i} \*u_{1i}^\top\dot{f}( \bm\alpha) \*x\right\|\\
&= \sup_{\| \*x\|=1, \*x\in R^{{q}}}\left(\sum_{i=1}^{{q}}\lambda_{1i}^2\| \*u_{1i}^\top\dot{f}( \bm\alpha) \*x\|^2\right)^{1/2}
\geq \sup_{\| \*x\|=1, \*x\in R^{{q}}}\frac{1}{\tilde{C}}\left(\sum_{i=1}^{{q}}\| \*u_{1i}^\top\dot{f}(\bm \alpha) \*x\|^2\right)^{1/2} \\
&= \sup_{\| \*x\|=1, \*x\in R^{{q}}}\frac{1}{\tilde{C}} \|\dot{f}( \bm\alpha) \*x\|=\frac{1}{\tilde{C}}\|\dot{f}( \bm\alpha)\|.
\end{split}
\end{equation}
Therefore, it follows from $\bm \alpha\in \*B_n$,  (\ref{superr:38}) and  (C4) that
\begin{align*}
\left\|\left\{ \*\Sigma_{n1}+(\lambda_n/n){\* D}_1( \bm\alpha)\right\}\dot{f}( \bm\alpha)\right\|
&\geq \left\| \*\Sigma_{n1}\dot{f}( \bm\alpha)\right\|-\left\|(\lambda_n/n){ \*D}_1( \bm\alpha)\dot{f}( \bm\alpha)\right\|\\
&\geq (1/\tilde{C})\|\dot{f}( \bm\alpha)\|-(\lambda_n/n)\cdot {a_{0}^{-2}}\|\dot{f}( \bm\alpha)\|,
\end{align*}
This, together with (\ref{eq:a}) and the fact  $\lambda_n/n\rightarrow 0$, implies that
\begin{equation}
\label{eq:b}
\sup_{ \bm\alpha\in \*B_n}\|\dot{f}( \bm\alpha)\|=o_p(1).
\end{equation}
Finally, we can get the conclusion in part (a) from (\ref{map}) and (\ref{eq:b}).

Next we prove part (b).  Write
\begin{equation}
\label{eq:decomposition}
\begin{split}
n^{1/2} \, {z^{-1} \*c^\top} (\hat{ \bm\alpha}^\circ- \bbeta_{01})&=n^{1/2} \, {z^{-1} \*c^\top} \left[ \left\{ \*\Sigma_{n1}+\frac{\lambda_n}{n} \*D_1(\hat{ \bm\alpha}^\circ)\right\}^{-1} \*\Sigma_{n1}- \*I_{q_n}\right] \bbeta_{01}\\
&+  n^{-1/2}\, {z^{-1} \*c^\top} \left\{ \*\Sigma_{n1}
+\frac{\lambda_n}{n} \*D_1(\hat{ \bm\alpha}^\circ)\right\}^{-1} \*X_{1}^\top {\bm\varepsilon}^* \equiv I_1 +I_2.
\end{split}
\end{equation}
By the first order resolvent expansion formula \[( \*H+ \*\Delta)^{-1}=  \*H^{-1}- \*H^{-1}\*\Delta( \*H+\*\Delta)^{-1},\]
the first term on the right hand side of equation (\ref{eq:decomposition}) can be rewritten as
\begin{equation}\nonumber
I_1 = - {z^{-1} \*c^\top}  \*\Sigma_{n1}^{-1}\frac{\lambda_n}{\sqrt{n}} \*D_1(\hat{ \bm\alpha}^\circ)\left\{ \*\Sigma_{n1}
+\frac{\lambda_n}{n} \*D_1(\hat{ \bm\alpha}^\circ)\right\}^{-1}\* \Sigma_{n1}\bbeta_{01}.
\end{equation}
Hence, by the assumption (C4) and (C5), we have
\begin{equation}\label{eq:27a}
	\|I_1\|\leq {\color{black}\frac{\lambda_n}{\sqrt{n}}{z^{-1}a_{0}^{-2}}\|\*\Sigma_{n1}^{-1}\bbeta_{01}\|=O_p\bigg({\lambda_n/\sqrt{n}}\bigg) \to  0.}
	\end{equation}
Furthermore, applying the first order resolvent expansion formula, it can be shown that
\begin{equation}\label{eq:A28}
\begin{split}
I_2&={\frac{z^{-1} }{\sqrt{n}}\*c^{\T} \*\Sigma_{n1}^{-1}\*X_1^{\T}{\bm\varepsilon}^*+o_p(1)}\\
&={\frac{z^{-1} }{\sqrt{n}}\*c^{\T} \*\Sigma_{n1}^{-1}\*X_1^{\T}(\*Y^*-\*X_1 \bm\mu+\*X_1 \bm\mu-\*X_1\bbeta_{01})+o_p(1)}\\
&={\sqrt{n}z^{-1}\*c^{\T}(\hat{\bbeta}_{1\rm Z}-\bm\mu_1+\bm\mu_1- \bbeta_{01}) +o_p(1)  }\\
\end{split}
\end{equation}
{where $\bm \mu_1 =(\mu_1^*, \mu_2^*, ..., \mu_{q}^*)$. $I_2$} converges in distribution to $N(0,1)$ by the Lindeberg-Feller central limit theorem. Finally, combining (\ref{eq:decomposition}), (\ref{eq:27a}), and (\ref{eq:A28}) proves part (b).
\end{proof}

\bigskip
\noindent
{\bf  Proof of Theorem \ref{theorem:CBAR}.}
Given the initial  ridge estimator $\hat{ \bbeta}^{(0)}$ in (\ref{ridge}), we have
\begin{equation}
\begin{split}
\hat{\bm\beta}^{(0)}-\bm\beta_0&=[ ( \*\Sigma_{n}+\frac{\xi_n}{n} \*I_{p_n})^{-1} \*\Sigma_{n}- \*I_{p_n}] \bbeta_{0}+( \*\Sigma_{n} +\frac{\xi_n}{n} \*I_{p_n})^{-1} \*X^\top {\bm\v}^*/n.\\
&\equiv \*T_1+\*T_2.
\end{split}
\end{equation}
By  the first order resolvent expansion formula and $\xi_n/\sqrt{n}\rightarrow 0$,
\begin{equation}
\|\*T_1\|=\left\|-  \*\Sigma_{n}^{-1}\frac{\xi_n}{{n}} ( \*\Sigma_{n}
+\frac{\xi_n}{n}\*I_{p_n})^{-1}\* \Sigma_{n}\bbeta_{0}\right\|\leq \tilde{C}^3\frac{\xi_n {a_{1}}\sqrt{p_n}}{n}
=o_p(\sqrt{\frac{p_n}{n}}).
\end{equation}
It is easy to see that $
\|\*T_2\|=O_p(\sqrt{{p_n}/{n}}).
$
Thus $\|\hat{ \bbeta}^{(0)}- \bbeta_0\|=O_p((p_n/n)^{1/2})$.
This, combined with part (a) of  Lemma 1, implies that
\begin{equation}\label{eq:A29}
\mbox{Pr}(\lim_{k\rightarrow\infty}{\hat{ \bm\gamma}^{(k)}}=  0)\rightarrow 1.
\end{equation}
Hence, to prove part (i) of Theorem 1, it is sufficient to show that
\begin{equation}\label{eq:A30}
\mbox{Pr}(\lim_{k\rightarrow \infty}\| {\hat{ \bm\alpha}^{(k)}}-\hat{ \bm\alpha}^\circ\|=0)\rightarrow 1,
\end{equation}
where $\hat{ \bm\alpha}^\circ$ is the fixed point of $f(\bm \alpha)$ defined in part (b) of Lemma 2.

Define $ \bm\gamma^*= 0$ if $\bm\gamma =  0$, for any $\bm \alpha\in \*B_n$,
\begin{equation}\label{eq:A31}
\lim_{ \bm\gamma\rightarrow 0}  \bm\gamma^*( \bm\alpha,\bm \gamma)= 0.
\end{equation}
Combining (\ref{eq:A31}) with the fact
\begin{equation}\nonumber
\begin{pmatrix}
 \*X_1^\top{ \*X_1}+\lambda_n \*D_1( \bm\alpha) &  \*X_1^\top{\* X_2}\\
 \*X_2^\top{ \*X_1}&  \*X_2^\top{ \*X_2}+\lambda_n \*D_2( \bm\gamma)
\end{pmatrix}
\begin{pmatrix}
 \bm\alpha^*\\
 \bm\gamma^*
\end{pmatrix}
=\begin{pmatrix}
 \*X_1^\top \*Y^*\\
 \*X_2^\top \*Y^*
\end{pmatrix},
\end{equation}
implies that for any $\bm \alpha\in \*B_n$,
\begin{equation}
\label{lim:111}
\lim_{ \bm\gamma\rightarrow 0} \bm\alpha^*(\bm \alpha, \bm\gamma)=\{ \*X_1^{\T}{\* X_1}+\lambda_n \*D_1( \bm\alpha)\}^{-1}\* X_1 \*Y^*{=f( \bm\alpha)}.
\end{equation}
Therefore, $g(\cdot)$ is continuous and thus uniformly continuous on the compact set $ \bbeta\in  \*H_n$.  This, together with (\ref{eq:A29}) and (\ref{lim:111}), implies that
 as $k\rightarrow\infty$,
\begin{eqnarray}
\label{eq:c}
\eta_k\equiv \sup_{\bm \alpha\in \*B_n}\left\|f( \bm\alpha)- \bm\alpha^*( \bm\alpha,{\hat{ \bm\gamma}^{(k)}})\right\|\longrightarrow 0,
\end{eqnarray}
with probability tending to 1.

Note that
\begin{eqnarray}
\nonumber
\| {\hat{ \bm\alpha}^{(k+1)}}-\hat{\bm \alpha}^\circ\| &=& \left\| \bm\alpha^*({\hat{ \bm\beta}^{(k)}})-\hat{ \bm\alpha}^\circ\right\|
\le  \left\| \bm\alpha^*({\hat{\bm \beta}^{(k)}})-f({\hat{\bm \alpha}^{(k)}})\right\|+\|f({\hat{ \bm\alpha}^{(k)}})-\hat{ \bm\alpha}^\circ\|\\
&\le & \eta_k + \frac{1}{\tilde{C}}\|{\hat{ \bm\alpha}^{(k)}}-\hat{ \bm\alpha}^\circ\|,
\label{eq:A34}
\end{eqnarray}
where the last step follows from
$\|f({\hat{ \bm\alpha}^{(k)}})-\hat{\bm \alpha}^\circ\|=\|f({\hat{ \bm\alpha}^{(k)}})-f(\hat{ \bm\alpha}^\circ)\|\leq (1/\tilde{C})\|{\hat{ \bm\alpha}^{(k)}}-\hat{ \bm\alpha}^\circ\|$.
Let $a_k=\| {\hat{ \bm\alpha}^{(k)}}-\hat{\bm \alpha}^\circ\|$, for all $k\geq 0$.
From (\ref{eq:c}), we can induce that  with probability tending to 1, for any $ \epsilon>0$, there exists an positive integer $N$ such that for all $k> N$, $|\eta_k|<\epsilon$ and
\begin{align*}
a_{k+1} &\leq \frac{a_{k-1}}{\tilde{C}^2} + \frac{\eta_{k-1}}{\tilde{C}}+\eta_k\\
& \leq \frac{a_1}{\tilde{C}^k}+\frac{\eta_1}{\tilde{C}^{k-1}}+ \cdots+\frac{\eta_N}{\tilde{C}^{k-N}}+ (\frac{\eta_{N+1}}{\tilde{C}^{k-N-1}}+\cdots +\frac{\eta_{k-1}}{\tilde{C}}+\eta_k)\\
&\le (a_1+\eta_1+...+\eta_N) \frac{1}{\tilde{C}^{k-N}} + \frac{1-(1/\tilde{C})^{k-N}}{1-1/\tilde{C}} \epsilon
\rightarrow 0, \mbox{ as } k\rightarrow\infty.
\end{align*}
This proves (\ref{eq:A30}).

Therefore, it immediately follows from (\ref{eq:A29}) and (\ref{eq:A30}) that the with probability tending to 1,
$ \lim_{k\to\infty} \bbeta^{(k)}= \lim_{k\to\infty} (\hat{ \bm\alpha}^{(k)\top} , \hat{ \bm\gamma}^{(k)\top})^\top=(\hat{ \bm\alpha}^{\circ\top} , 0)^{\T}$, which completes the proof of
part (i). This, in addition to  part (b) of Lemma 2, proves part (ii) of Theorem \ref{theorem:CBAR}.   \qed

{{\bf  Proof of Theorem \ref{theorem:2}.}
Recall that $\hat{\bbeta}^* =\lim_{k\to\infty}\hat{ \bm\beta}^{(k+1)}$ and
$\hat{ \bm\beta}^{(k+1)}=\arg\min_{  \bm\beta} \{ Q(\bbeta| \hat{ \bm\beta}^{(k)})\}$, where
$$
Q(\bbeta| \hat{ \bm\beta}^{(k)})= \| \*Y^*-\* X \bbeta\|^2+\lambda_n
 \sum_{\ell=1}^{p_n} {\beta_\ell^2}/{\{\hat{\beta}^{(k)}_\ell\}^2}.
$$
If $\beta_\ell^*\ne 0$ for $\ell \in \{ i,j\}$, then $\hat{\bbeta}^* $ must
satisfy the following normal equations for $\ell \in \{ i,j\}$:
\begin{equation}\nonumber
-2 \*x_\ell^\top \{\*Y^*- \*X\hat{ \bbeta}^{(k+1)}\}+2\lambda_n {\hat{\beta}_\ell^{(k+1)}}/{\{\hat{\beta}^{(k)}_\ell\}^2} =0.
\end{equation}
Thus, for $\ell \in \{ i, j \}$,
\begin{equation}
\label{eqss:2}
{\hat{\beta}_\ell^{(k+1)}}/{\{\hat{\beta}^{(k)}_\ell \}^2}  = {\*x_\ell^\top \hat{\bm\varepsilon}^{*(k+1)}}/{\lambda_n},
\end{equation}
where $\hat{\bm\varepsilon}^{*(k+1)}= \*Y^*- \*X\hat{ \bbeta}^{(k+1)}$.
Moreover, because
$$\|\hat{\bm\varepsilon}^{*(k+1)}\|^2+ \lambda_n\sum_{i=1}^{p_n}\frac{\hat{\beta}_i^{2}}{\tilde{\beta}_i^2}=Q(\hat{ \bm\beta}^{(k+1)}| \hat{ \bm\beta}^{(k)})\leq
Q(0|\hat{ \bm\beta}^{(k)}) =\| \*Y^*\|^2,$$
we have
\begin{equation}
\label{eq:A36}
\|\hat{\bm\varepsilon}^{*(k+1)}\|\leq \| \*\*Y^*\|
\end{equation}
Letting $k\to \infty$ in (\ref{eqss:2}) and (\ref{eq:A36}), we have, for $\ell \in \{i, j\}$ and $\|\hat{\bm\varepsilon}^{*}\|\leq \| \*Y^*\|$,
$\hat{\beta}_\ell ^{*-1}= \*x_\ell ^\top \hat{\bm\varepsilon}^{*}{\lambda_n}$,
where $\hat{\bm\varepsilon}^{*} = \*Y^*- \*X\hat{ \bbeta}^*$.
Therefore,
\[\big|\hat{\beta}_i^{*-1}-\hat{\beta}_j^{*-1}\big|\leq \frac{1}{\lambda_n} \, \| \*Y^*\| \times \|\*x_i - \*x_j\| = \frac{1}{\lambda_n}
\, \| \*Y^*\|\sqrt{2(1-\rho_{ij})}.
\]
 \hfill $\Box$}

\end{document}